\documentclass{article}
\usepackage[a4paper, margin=1in]{geometry}
\usepackage[parfill]{parskip}
\usepackage[T1]{fontenc}
\usepackage[bookmarksnumbered,bookmarksopen,unicode,colorlinks,allcolors=blue]{hyperref}
\usepackage{url}
\usepackage{booktabs}
\usepackage{amsfonts}
\usepackage{nicefrac}
\usepackage{microtype}
\usepackage{xcolor}
\usepackage{authblk}
\usepackage{amsmath}
\usepackage{amsthm}
\usepackage{amssymb}
\usepackage{graphicx}
\usepackage{complexity}
\usepackage{algorithm}
\usepackage[noend]{algpseudocode}
\algrenewcommand\algorithmicrequire{\textbf{Input:}}
\algrenewcommand\algorithmicensure{\textbf{Output:}}
\usepackage{thmtools}
\usepackage{cleveref}
\usepackage[numbers]{natbib}
\usepackage{todonotes}

\newtheorem{theorem}{Theorem}[section]
\newtheorem{definition}{Definition}
\newtheorem{corollary}[theorem]{Corollary}

\newtheorem{lemma}[theorem]{Lemma}

\newtheorem{example}[theorem]{Example}

\newcommand{\RR}{\mathbb{R}}

\title{The Computational Complexity of Counting Linear Regions\\in ReLU Neural Networks \thanks{This paper is an extension of a preliminary conference version that appeared in the proceedings of the NeurIPS 2025
conference \cite{stargalla2026computational}.}}

\author[1]{Moritz Stargalla}
\author[1]{Christoph Hertrich}
\author[2]{Daniel Reichman}
\affil[1]{\normalsize University of Technology Nuremberg}
\affil[2]{\normalsize Worcester Polytechnic Institute}
\date{}

\renewcommand{\paragraph}[1]{\textbf{#1}.}

\begin{document}

\maketitle

\begin{abstract}
\noindent
An established measure of the expressive power of a given ReLU neural network is the number of linear regions into which it partitions the input space.
There exist many different, non-equivalent definitions of what a linear region actually is. We systematically assess which papers use which definitions and discuss how they relate to each other. We then analyze the computational complexity of counting the number of such regions for the various definitions. Generally, this turns out to be an intractable problem. We prove \NP{}- and \#\P{}-hardness results already for networks with one hidden layer and strong hardness of approximation results for two or more hidden layers. Finally, on the algorithmic side, we demonstrate that counting linear regions can at least be achieved in polynomial \emph{space} for some common definitions.
\end{abstract}

\section{Introduction}

Neural networks with rectified linear unit (ReLU) activations are among the most common and fundamental models in modern machine learning. The functions represented by ReLU networks are \emph{continuous and piecewise linear} (CPWL), meaning that the input space can be partitioned into finitely many pieces on each of which the function is affine. Such pieces are called \emph{linear regions}. This leads to the following intuition: the more linear regions a neural network can produce, the more complex problems it is capable of solving. Consequently, starting with \cite{pascanu2013number} and \cite{montufar2014number}, the number of linear regions became a standard measure of the expressive power of a ReLU network. Substantial effort has been put into understanding this quantity, e.g., by deriving upper and lower bounds depending on the network architecture or by developing algorithms to count it. More information can be found in the surveys \cite{balestriero2025geometry} and \cite{huchette2023deep}.

Despite the significant interest in understanding the number of linear regions, surprisingly little is known about the most natural associated computational complexity question: Given a neural network,
what are the time and space requirements needed to determine how many regions it has? The main objective of our paper is to make progress on this question by proving 
complexity-theoretical results on the problem of counting linear regions.

However, before one can even talk about counting linear regions, one has to properly define them. What sounds like a simple exercise is actually a non-trivial task. In the literature, there exists a variety of non-equivalent definitions of what counts as a linear region of a ReLU network. For example, some authors define it via possible sets of active neurons, others define it solely based on the function represented by the neural network. Some authors require regions to be full-dimensional, or connected, or even convex, others do not.
Inconsistencies between definitions have led to confusion and even minor flaws in previous work, as we explain in \Cref{sec:flaws}.

\subsection{Our Contributions}\label{sec:contribution}

\paragraph{Definitions of linear regions}
In order to raise awareness to the technical, but important and non-trivial inconsistencies regarding the definition of linear regions in neural networks, we identify six non-equivalent, commonly used definitions in \Cref{sec:definitions}. We discuss how they relate to each other and provide a table demonstrating which authors used which definitions in previous work. We do not make a recommendation about what definition is the most reasonable one to use, as this depends on the context, but we encourage all authors of future papers to be aware of the subtleties carried by the different options and to be explicit about which definition they use and why. 

\paragraph{Complexity of counting regions in shallow networks} As for many questions regarding ReLU networks, it makes sense to first understand the most basic case with one hidden layer. In \Cref{sec:shallow}, we prove that, regardless of which of the six definitions one uses, the seemingly simple question of deciding whether a shallow network has more than one linear region can indeed be decided in polynomial time. 
However, for all six definitions, we show that determining the exact number of regions is \#\P-hard, meaning that, unless the commonly believed conjecture $\#\P\,\neq\,\FP$ fails, one cannot count regions of a shallow network in polynomial time. 
Furthermore, our reduction shows that even finding an algorithm that approximately counts the number of regions for one hidden layer might be intractable, as it would resolve long-standing open questions in the context of counting cells of hyperplane arrangements \cite{linial1986hard}.

\paragraph{Complexity of counting regions in networks with more than one hidden layer}
\citet{wang2022estimation} showed that deciding if a deep neural network has more than $K$ regions is NP-hard.
In \Cref{sec:deep} we improve upon this result in several aspects. While the hardness by \citet{wang2022estimation} only applies to networks with logarithmically growing depth (in the input dimension), we show that hardness can be proved for every constant number of hidden layers $\geq 2$ and even in the case $K=1$, that is, for deciding if the network has more than one linear region. Our reduction also implies running-time lower bounds based on the exponential-time hypothesis. We furthermore show that, unless common complexity assumptions fail, one cannot even approximate the number of regions within an exponential factor in polynomial time.

\paragraph{Counting regions using polynomial space}
While most of our results are concerned with lower bounds, in \Cref{sec:pspace}, we turn our attention towards proving an \emph{upper} bound on the computational complexity of region counting. \citet{wang2022estimation} proved\footnote{The proof by \citet{wang2022estimation} works for a different definition than claimed in their paper; see \Cref{sec:flaws}.} that for one definition of linear regions, the problem can be solved in exponential time. We show the stronger statement that for three of our definitions, polynomial space is sufficient.

\subsection{Related work}
\citet{huchette2023deep} survey polyhedral methods for deep learning, also treating the study of linear regions in detail. To the best of our knowledge, the first bounds on the number of regions in terms of the network architecture (e.g., number of neurons, network depth) were developed by \citet{pascanu2013number} and \citet{montufar2014number}. Subsequently, better bounds were established \cite{raghu2017expressive,arora2018understanding,serra2018bounding,zanotti2025bounds}. \citet{arora2018understanding} prove that every CPWL function can be represented by a ReLU network.

Several works have developed algorithms for enumerating or exactly counting linear regions.
\citet{xu2022traversing,liu2023relu,vincent2025reachable}, and \citet{gaines2026characterizing} present algorithms for enumerating linear regions that start from one linear region and then successively enumerate all neighboring linear regions.
\citet{balestriero2023} present an enumeration algorithm based on recursively computing cells of hyperplane arrangements.
\citet{masden2022algorithmic} presents an algorithm to enumerate the full combinatorial structure of activation regions.
\citet{piwek2023exact} give an algorithm, based on tropical geometry, that counts the exact number of regions.
\citet{lezeau2024tropical} give an algorithm that returns explicit unions of polyhedra corresponding to the closures of open connected regions.
\citet{joyce2025computing} present an algorithm for networks with skip connections.
\citet{serra2018bounding} and \citet{cai2023getting} present mixed-integer programming based routines to count the number of regions in a bounded domain.
\citet{humayun2023splinecam,plate2026analysis}, and \citet{wei2026affinelens} also present algorithms that compute all linear regions in a bounded domain.
As discussed above, \citet{wang2022estimation} provides some initial results on the computational complexity of counting regions, which we strengthen significantly in this paper. Our reductions are related to other decision problems on trained neural networks such as verification \cite{katz2017reluplex,froese2024complexity,froese2025parameterized} or deciding whether the Lipschitz constant of a ReLU network exceeds a certain threshold \cite{virmaux2018lipschitz,jordan2020exactly}.

Another line of research has studied the question of how to construct ReLU networks for functions with a certain number of regions \cite{he2018relu,chen2022improved,hertrich2023towards,brandenburg2025decomposition,zanotti2025linear}. The number of regions of maxout networks was studied by \citet{montufar2022sharp}. Note that all our hardness results hold for maxout networks, too, as maxout is a generalization of ReLU. \citet{goujon2024number} present bounds for general piecewise linear activation functions. The average number of linear regions was studied, among others, by \citet{hanin2019deep} and \citet{tseran2021expected}. Our work is inspired by the aim to better understand complexity-theoretic aspects of neural networks; another well-studied question in that regime is the complexity of training \cite{goel2021tight,froese2022computational,froese2023training,bertschinger2023training}.

\section{Preliminaries}
\label{sec:prelim}
For $n\in \mathbb{N}$, we write $[n]:=\{1,\dots, n\}$.
For a set $P\subseteq \mathbb{R}^n$, we denote by $\overline{P}$, $P^\circ$, and $\partial P$ its closure, interior, and boundary, respectively.
A set $P$ in $\RR^n$ is \emph{regular open} if it is the interior of its closure, that is, $P=(\overline{P})^\circ$, and \emph{connected} if there are no two nonempty disjoint subsets $Q,R\subseteq P$ with $Q\cup R=P$, $\overline{Q}\cap R=\emptyset$, and $Q\cap \overline{R}=\emptyset$.
The ReLU function is the real function $x\mapsto \max(0,x)$.
For any $n\in \mathbb{N}$, we denote by $\sigma:\mathbb{R}^n\to \mathbb{R}^n$ the function that computes the ReLU function in each component.

\paragraph{Polyhedra, CPWL functions, and hyperplane arrangements}
A \emph{polyhedron} $P$ is the intersection of finitely many closed halfspaces.
A \emph{polytope} is a bounded polyhedron.
A \emph{face} of $P$ is either the empty set or a set of the form $\arg\min \{c^\top x: x\in P\}$ for some $c\in \mathbb{R}^n$.
A \emph{polyhedral complex} $\mathcal{P}$ is a finite collection of polyhedra such that $\emptyset \in \mathcal{P}$, if $P\in \mathcal{P}$ then all faces of $P$ are in $\mathcal{P}$, and if $P,P'\in \mathcal{P}$, then $P\cap P'$ is a face of $P$ and $P'$.
A function $f:\mathbb{R}^n\to \mathbb{R}$ is \emph{continuous piecewise linear} (CPWL), if there exists a polyhedral complex $\mathcal{P}$ such that the restriction of $f$ to each full-dimensional polyhedron $P\in \mathcal{P}$ is an affine function.
If this condition is satisfied, then $f$ and $\mathcal{P}$ are \emph{compatible}.
A vector $x\in \RR^n$ is a \emph{breakpoint} of $f$ if there is no open set $U\subseteq \RR^n$ containing $x$ such that $f$ is affine on $U$.
A \emph{hyperplane arrangement} $\mathcal{H}$ is a collection of hyperplanes in $\mathbb{R}^n$.
A \emph{cell} of a hyperplane arrangement is an inclusion maximal connected subset of $\mathbb{R}^n\setminus(\bigcup_{H\in \mathcal{H}} H)$.
A hyperplane arrangement naturally induces an associated polyhedral complex with the cells being the maximal polyhedra of the complex.

\paragraph{ReLU networks}
A \emph{ReLU neural network} $N$ with $d\geq 0$ hidden layers is defined by $d+1$ affine transformations $T^{(\ell)}:\mathbb{R}^{n_{\ell-1}}\to \mathbb{R}^{n_\ell}, x\mapsto A^{(\ell)}x+b^{(\ell)}$ for $\ell\in [d+1]$.
We assume that $n_0=n$ and $n_{d+1}=1$.
A ReLU network $N$ \emph{computes} the CPWL function $f_N:\mathbb{R}^n\to \mathbb{R}$ with
\[
f_N=T^{(d+1)}\circ \sigma \circ \cdots \circ \sigma \circ T^{(1)}.
\]
The matrices $A^{(\ell)}\in \mathbb{R}^{n_\ell\times n_{\ell-1}}$ are called the \emph{weights} and the vectors $b^{(\ell)}\in \mathbb{R}^{n_\ell}$ are the \emph{biases} of the $\ell$-th layer.
We say the network has \emph{depth} $d+1$ and \emph{size} $s(N) := \sum_{\ell=1}^d n_\ell$.
We define $n_{\max}=\max\{n_0,n_1,\dots,n_{d+1}\}$.
Equivalently, ReLU networks can also be represented as layered, directed, acyclic graphs where each dimension of each layer is represented by one vertex, called a \emph{neuron}.
Each neuron computes an affine transformation of the outputs of its predecessors, applies the ReLU function, and outputs the result.
We denote the CPWL function mapping the network input to the output of a neuron $i$ by $f_{N,i}:\mathbb{R}^n\to \mathbb{R}$.
If the reference to the ReLU network $N$ is clear, we abbreviate $f_{N,i}$ by $f_i$.

\paragraph{Activation patterns}
Given a ReLU network $N$, a vector $a \in \{0,1\}^{s(N)}$ is called an \emph{activation pattern} of $N$ if there exists an input $x \in \mathbb{R}^n$ such that when $N$ receives $x$ as input, the $i$-th neuron in $N$ has positive output (is active) if $a_i=1$ and output zero if $a_i=0$.
For any vector $a\in \{0,1\}^{s(N)}$, not necessarily an activation pattern, we define affine preactivation functions $g^a_{N,i}:\RR^n\to \RR$ and affine output functions $f^a_{N,i}:\RR^n\to \RR$ for every neuron as follows.
For each hidden layer $\ell\in [d]$, let $D_a^{(\ell)}\in \RR^{n_\ell\times n_\ell}$ be the diagonal matrix whose $j$-th diagonal entry is $1$ if the $j$-th neuron of layer $\ell$ is active in $a$, and $0$ otherwise.
If $g^{(\ell),a}:\RR^n\to \RR^{n_\ell}$ and $f^{(\ell),a}:\RR^n\to \RR^{n_\ell}$ denote the vector of preactivation and output functions in layer $\ell$, then $g^{(1),a}(x)=A^{(1)}x+b^{(1)}$ and $f^{(1),a}(x)=D_a^{(1)}g^{(1),a}(x)$, and for $\ell=2,\dots, d$, $g^{(\ell),a}(x)=A^{(\ell)}f^{(\ell-1),a}(x)+b^{(\ell)}$ and $f^{(\ell),a}(x)=D_a^{(\ell)}g^{(\ell),a}(x)$.
The collapsed network output is then the affine function $f_N^a(x)=A^{(d+1)}f^{(d),a}(x)+b^{(d+1)}$.
If neuron $i$ is the $j$-th neuron in layer $\ell$, then $g^a_{N,i}$ and $f^a_{N,i}$ are the $j$-th components of $g^{(\ell),a}$ and $f^{(\ell),a}$
In particular, $g_{N,i}^a$ is the affine function at the input of neuron $i$, and $f_{N,i}^a$ is the affine function at the output of neuron $i$.
In particular, $f_{N,i}^a=g_{N,i}^a$ if $a_i=1$, and $f_{N,i}^a\equiv 0$ if $a_i=0$.
Again, if the reference to $N$ is clear, we abbreviate $g^a_{N,i}$ and $f^a_{N,i}$ by $g^a_i$ and $f^a_i$.

\paragraph{Encoding size}
We use $\langle \cdot \rangle$ to denote the encoding size of numbers, matrices, or entire neural networks, where we assume that numbers are integers or rationals encoded in binary such that they take logarithmic space.
The \emph{encoding size} $\langle n\rangle$ of a nonnegative integer $n$ is $\langle n\rangle:=\lceil\log_2(n+1)\rceil$, the encoding size of a fraction $q=a/b$ with $a\in \mathbb{Z}$ and $b\in \mathbb{N}$ is $\langle q\rangle := 1 + \langle |a|\rangle + \langle b\rangle$, and the encoding size of a matrix with rational entries $A=(a_{ij})_{i\in [n],\ j\in [m]}$ is $\langle A\rangle := nm + \sum_{i\in [n],\ j\in [m]}\langle a_{ij}\rangle$.
In the following, we restrict ourselves to ReLU networks with rational entries that can be represented with a finite number of bits.
Then, a ReLU network $N$ of depth $d+1$ with $A^{(i)}\in \mathbb{Q}^{n_i\times n_{i-1}}$ and a vector $b^{(i)}\in \mathbb{Q}^{n_i}$ has encoding size
$
\langle N\rangle = \sum_{i=1}^{d+1} \langle A^{(i)}\rangle + \sum_{i=1}^{d+1}\langle b^{(i)}\rangle
$.
In particular, if $A_{\max}$ denotes the maximum absolute value of a numerator or denominator in an entry in $A^{(i)}$ and $b^{(i)}$, then $\langle N\rangle =\text{poly}(d,n_{\max},\langle A_{\max}\rangle)$.

\paragraph{Computational complexity}
We give an informal overview of some notions of computational complexity and refer to \citet{arora2009computational} for further reading.
A function $f:\{0,1\}^*\to \{0,1\}$ is in \P\ if $f$ is computable in polynomial time by a deterministic Turing machine, in \NP\ if it is computable in polynomial time by a non-deterministic Turing machine, and in \RP\ if it is computable in polynomial time by a randomized Turing machine that never outputs false positives and accepts a correct input with probability at least $1/2$.
Intuitively, \P\ contains problems that can be efficiently solved while \NP\ contains those whose solutions can be efficiently verified.
It is widely believed that \P\,$\neq$\,\NP\ and \RP\,$\neq$\,\NP\ hold.
A function $f:\{0,1\}^*\to \mathbb{N}$ is in \FP\ if $f$ is computable in polynomial time by a deterministic Turing machine, in \#\P\ if there is a polynomial time non-deterministic Turing machine which has exactly $f(x)$ accepting paths for any input $x\in \{0,1\}^*$, and in \FP\SPACE{} if $f$ is computable by a deterministic Turing machine that uses polynomial space.
A problem is called \emph{hard} for \NP\ (analogously, for \#\P) if all other problems in this class can be reduced to it in polynomial time, and \emph{complete} if it is both hard and contained in the class itself.

\section{Definitions of linear regions}\label{sec:definitions}

In this section we first extract the six most commonly used definitions of linear regions from the literature and discuss their relations together with important properties and subtleties.
We also discuss some prior works where the chosen definition of `linear regions' led to inconsistencies, and give suggestions how to fix them.
\Cref{tab:literature} provides an overview of which previous papers use which definitions.

\begin{table}[ht!]
    \centering
    \hfill
    \begin{tabular}{ll}
        \toprule
         Paper & Definitions \\
         \midrule
         \citet{montufar2014number} (2014) & \ref{def:region5} \\
         \citet{pascanu2013number} (2014) & \ref{def:region4} \\
         \citet{raghu2017expressive} (2017) & \ref{def:region1}, \ref{def:region5} (*) \\
         \citet{arora2018understanding} (2018) & \ref{def:region5} \\
         \citet{serra2018bounding} (2018) & \ref{def:region1}, \ref{def:region2} (*)\\
         \citet{hanin2019deep} (2019) & \ref{def:region2}, \ref{def:region4} \\
         \citet{he2018relu} (2020) & \ref{def:region3}, \ref{def:region6} \\
         \citet{rolnick2020reverse} (2020) & \ref{def:region2}, \ref{def:region5} (*) \\
         \citet{tseran2021expected} (2021) & \ref{def:region1}, \ref{def:region5} (*)\\
         \citet{chen2022improved} (2022) & \ref{def:region3}, \ref{def:region5}, \ref{def:region6} \\
         \citet{montufar2022sharp} (2022) & \ref{def:region4} \\
        \bottomrule
    \end{tabular}\hfill\hfill
    \begin{tabular}{ll}
        \toprule
         Paper & Definitions \\
         \midrule
         \citet{wang2022estimation} (2022) & \ref{def:region1}, \ref{def:region5}  \\
         \citet{cai2023getting} (2023) & \ref{def:region2}  \\
         \citet{hertrich2023towards} (2023) & \ref{def:region6}  \\
         \citet{goujon2024number} (2024) & \ref{def:region3}, \ref{def:region6} \\
         \citet{masden2022algorithmic} (2025) & \ref{def:region2} \\
         \citet{zanotti2025bounds} (2025) & \ref{def:region4}, \ref{def:region6} \\
         \citet{zanotti2025linear} (2025) & \ref{def:region4} (*)\\
         \citet{brandenburg2025decomposition} (2026) & \ref{def:region3}, \ref{def:region6} \\
         \citet{gaines2026characterizing} (2026) & \ref{def:region2}\\
         \citet{huchette2023deep} (2026) & \ref{def:region2} (*)  \\
         \citet{lezeau2024tropical} (2026)& \ref{def:region4} (*)\\
        \bottomrule
    \end{tabular}
    \hfill{}\vspace{.3em}
    \caption{
    List of papers that use one or several definitions.
    Additional notes on the papers marked with an asterisk can be found in the footnote\protect\footnotemark.
    }
    \label{tab:literature}
\end{table}
\footnotetext{
    \citet{raghu2017expressive} do not use \Cref{def:region1} explicitly, but they define activation patterns and derive a bound on the total number of activation patterns, which is equivalent to bounding the number of activation regions.
    The bound of \citet{serra2018bounding} holds for \Cref{def:region2}, compare the discussion in \citep[Section 5]{cai2023getting}. The mixed-integer program counts the number of activation regions (\Cref{def:region1}).
    \citet{rolnick2020reverse} only treat cases where \Cref{def:region2,def:region5} are equivalent.
    \citet{tseran2021expected} consider activation regions of maxout networks, which is conceptually slightly different from the activation regions defined in this paper.
    \citet{huchette2023deep} actually define activation regions using \Cref{def:region1}, but then later state that they disregard low dimensional linear regions. Therefore, their definition is equivalent to \Cref{def:region2}.
    \citet{zanotti2025linear} defines linear regions as the closure of open connected regions.
    \citet{lezeau2024tropical} define linear regions as closures of open connected regions.
}

The set of inputs that have the same activation pattern induces a subset of $\mathbb{R}^n$ on which $f_N$ is affine.
\begin{definition}[Activation Region]
\label{def:region1}
    Given a network $N$ and an activation pattern $a\in\{0,1\}^{s(N)}$ with support $I\subseteq [s(N)]$, the set
    $
    S_{N,a}=\{x\in \mathbb{R}^n: g_i^{a}(x)>0\; \text{ for all } i\in I,\; g_i^{a}(x) \leq 0\; \text{ for all } i\notin I\}
    $
    is an \emph{activation region} of $N$.
    If the reference to the ReLU network is clear, we abbreviate $S_{N,a}$ by $S_a$.
\end{definition}
The inequalities in the definition of $S_a$ ensure that the activation pattern of every $x\in S_a$ is $a$.
More generally, for every vector $a\in \{0,1\}^{s(N)}$, not necessarily an activation pattern, we claim that
\[
S_a=\{x\in \RR^n: \text{the activation pattern of $N$ at $x$ is $a$}\}.
\]
To see this, observe that in the first hidden layer, the actual preactivation functions are exactly $g_i^a$, independently of $a$.
Thus, the inequalities $g_i^a(x)>0$ for $a_i=1$ and $g^a_i(x)\leq 0$ for $a_i=0$ are exactly the conditions for the neurons in the first layer to have the states prescribed by $a$.
The actual outputs of the neurons are then the values of $f_i^a$.
Inductively, suppose this holds up to layer $\ell$.
Then, the actual inputs to the neurons in the next layer are precisely the affine functions $g_i^a(x)$, since they are obtained by applying the next affine transformation to the outputs $f^a_j(x)$ of the previous layer.
Hence, the inequalities defining $S_a$ again state exactly that the neurons in this layer have the states prescribed by $a$, which proves the claim.
In particular, $a$ is an activation pattern if and only if $S_a\neq \emptyset$.

\begin{figure}[t]
  \centering
  \raisebox{-0.5\height}{
	\includegraphics[page=1, width=0.7\textwidth]{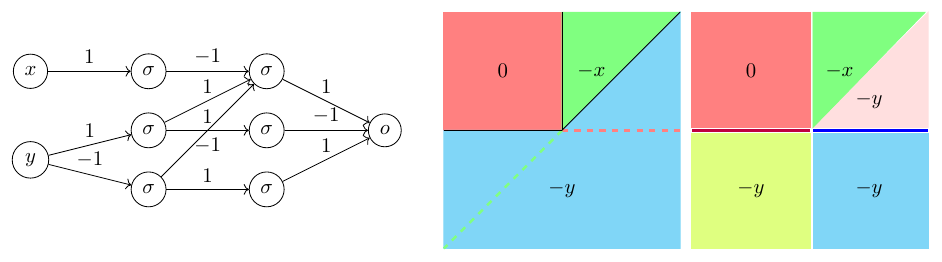}
    }
    \raisebox{-0.5\height}{
	\includegraphics[page=2, width=0.26\textwidth]{arxiv_figures.pdf}
    }
  \caption{
        A ReLU network (left) computing the function $f(x,y)=\max(-y, \min(0,-x))$ (right).
        The closed connected regions (center left) and the activation regions (center right) are displayed.
        The slice $\{(x,0):x\geq 0\}$ is contained in the two closed connected regions with functions $0$ (red) and $-y$ (blue).
		We have $R_6=R_5=R_4=3$, $R_3=4$, $R_2=5$ and $R_1=7$.
		In anti-clockwise direction starting from the region with value 0, the activation patterns are $010110$, $000000$, $001001$, $101001$, $100000$, $110010$ and $110110$ (neurons are ordered from the upper left to the lower right).}
		\label{fig:activation_regions}
\end{figure}

We illustrate this with a small example.
\begin{example}\label{ex:activation_region_example}
    Consider the network in \Cref{fig:activation_regions} where the neurons are numbered from the upper left to the lower right.
    We consider the vector $a=(1,1,0,1,1,0)$ and show that $S_a$ is exactly the set of inputs whose activation pattern is $a$.
    All biases of the network are zero, and the weight matrices are
    \[
    A^{(1)}=\begin{pmatrix}1&0\\0&1\\0&-1\\\end{pmatrix},
    \quad
    A^{(2)}=\begin{pmatrix}-1&1&-1\\0&1&0\\0&0&1\\\end{pmatrix},
    \quad
    A^{(3)}=\begin{pmatrix}1&-1&1\\\end{pmatrix}.
    \]
    The activation pattern $a$ describes that the first two neurons of each hidden layer are active and the third one is inactive.
    Hence $D_a^{(1)}=D_a^{(2)}=\begin{pmatrix}e_1&e_2&0\\\end{pmatrix}$, where $e_1=(1,0,0)^\top$ and $e_2=(0,1,0)^\top$.
    The preactivation functions of the first hidden layer are
    \[
    g^{(1),a}(x,y)=A^{(1)}\begin{pmatrix}x\\y\end{pmatrix}=\begin{pmatrix}x\\y\\-y\\\end{pmatrix}.
    \]
    Thus, the first three entries $110$ of $a$ require $x>0$, $y>0$, and $-y\leq 0$.
    Under these inequalities, the outputs of the first hidden layer are
    \[
    f^{(1),a}(x,y)=D_a^{(1)}g^{(1),a}(x,y)=\begin{pmatrix}x\\y\\0\\\end{pmatrix}.
    \]
    Using these output functions as inputs for the second hidden layer gives
    \[
    g^{(2),a}(x,y)=A^{(2)}\begin{pmatrix}x\\y\\0\\\end{pmatrix}=\begin{pmatrix}-x+y\\y\\0\\\end{pmatrix}
    \quad\text{ and }\quad
    f^{(2),a}(x,y)=D_a^{(2)}g^{(2),a}(x,y)=\begin{pmatrix}-x+y\\y\\0\\\end{pmatrix},
    \]
    where the last three entries $110$ of $a$ require $-x+y\geq 0$, $y>0$, and $0\leq 0$.
    Combining the conditions from both layers gives
    \[
    S_a=\{(x,y)\in \RR^2: x>0,y>0,-y\leq 0, -x+y>0,y>0,0\leq 0\}
    =\{(x,y)\in \RR^2: x>0,y>x\}.
    \]
    In particular, if $(x,y)\in S_a$, then the signs of the three first-layer preactivations are $+,+,-$, and the signs of the second-layer preactivations are also $+,+,-$, implying that $(x,y)$ has activation pattern $a$.
    Thus $S_a$ is exactly the set of inputs with activation pattern $a$.
    Finally, the output of the network on $S_a$ is $f^a_N(x,y)=A^{(3)}f^{(2),a}(x,y)=(-x+y)-y=-x$, that is, $f_N(x,y)=-x$ for all $(x,y)\in S_a$.
\end{example}

Activation regions can be open, closed, neither open nor closed, and full- or low-dimensional, see \Cref{fig:activation_regions} for some examples.
The (disjoint) union of all activation regions is exactly $\mathbb{R}^n$.
In particular, the number of activation regions equals the number of activation patterns.
It is important to note that the term \emph{activation region} is used ambiguously.
For example, \citet{hanin2019deep} use the term to refer to only \emph{full-dimensional} activation regions. 
\begin{definition}[Proper Activation Region]
\label{def:region2}
Given a ReLU network $N$, a \emph{proper activation region} of $N$ is a full-dimensional activation region of $N$.
\end{definition}

While the previous two definitions depend on the neural network \emph{representation} itself, the following four definitions depend only on the CPWL \emph{function} represented by the ReLU network and are independent from the concrete representation.

\begin{definition}[Convex Region]
\label{def:region3}
    Given a ReLU network $N$ and a polyhedral complex $\mathcal{P}$ that is compatible with $f_N$, a \emph{convex region of $N$ given $\mathcal{P}$} is a full-dimensional polyhedron $P\in \mathcal{P}$.
    The \emph{number of convex regions} of $N$ is the minimum number of convex regions of any polyhedral complex $\mathcal{P}$ that is compatible with $f_N$.
\end{definition}
Note that many different polyhedral complexes can attain this minimal number (see, for example,  \Cref{fig:activation_regions}).
Hence, in general, it is not possible to refer to a polyhedron $P$ as a `convex region of $N$' without specifying an associated polyhedral complex.

Another option to define linear regions is to use inclusion-maximal connected subsets on which the function computed by the ReLU network is affine, leading to the following definitions.
\begin{definition}[Open Connected Region]
\label{def:region4}
    Given a ReLU network $N$, an \emph{open connected region} of $N$ is an open, inclusion-wise maximal connected subset of $\mathbb{R}^n$ on which $f_N$ is affine.
\end{definition}

\begin{definition}[Closed Connected Region]
\label{def:region5}
    Given a ReLU network $N$, a \emph{closed connected region} of $N$ is a full-dimensional closed inclusion-wise maximal connected subset of $\mathbb{R}^n$ on which $f_N$ is affine.
\end{definition}

The subtle difference in the definition of open and closed connected regions has an important consequence:
As \citet{zanotti2025bounds} showed, $\overline{P_1}\cap\overline{P_2}=\partial P_1\cap \partial P_2$ holds for any distinct open connected regions $P_1,P_2$.
Interestingly, the same is \emph{not true} for closed connected regions. This is because a closed connected region can continue on a low-dimensional slice of another closed connected region, which leads to a part of the boundary of one closed connected region to be contained in the interior of another closed connected region.
\citet[Figure 1]{zanotti2025bounds} (see \Cref{fig:zanotti_network}) gives a neat example where a low dimensional slice even connects two seemingly disconnected full-dimensional sets; another example can be found in \Cref{fig:activation_regions}.

\citet{hanin2019deep} define the set of open connected regions as the connected components of the input space where the breakpoints of $f_N$ are removed and prove that every open connected region is the interior of the closure of a union of some proper activation regions \citep[Lemma 3]{hanin2019deep}.
In contrast, a closed connected region is in general not the closure of a union of some activation regions, which is illustrated by the green region in \Cref{fig:activation_regions} and the orange region in \Cref{fig:zanotti_network}.

\begin{figure}
	\centering
	\includegraphics[page=17, width=0.4\textwidth]{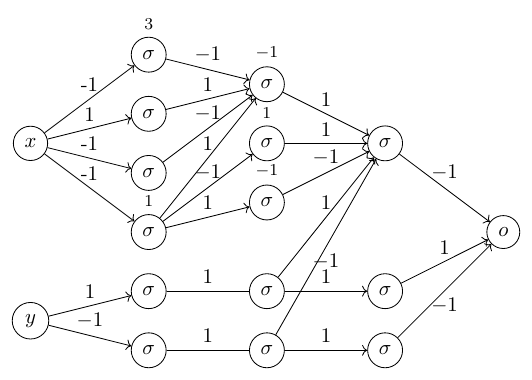}
	\includegraphics[page=18, width=0.26\textwidth]{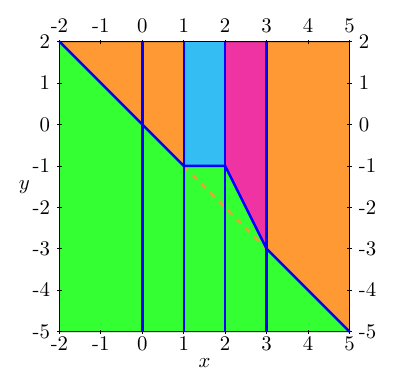}
	\includegraphics[page=20, width=0.31\textwidth]{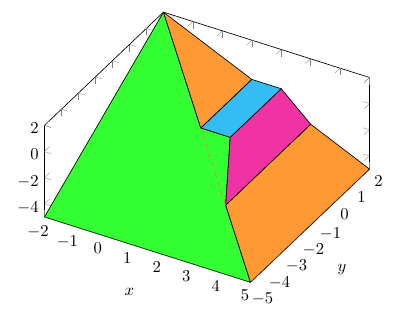}
	\caption{
        A ReLU network $N$ computing $\min(y, \max(-1, -x), \max(3 - 2 x, -x))$ (this function was first considered by \citet[Figure 1]{zanotti2025bounds}).
		An activation region of $N$ is either a blue line, blue point, or a full dimensional cell as defined by the blue lines.
		There are four closed connected regions as indicated by the colors. The line between the points $(1,-1)$ and $(3,-3)$ belongs to the green as well as the orange region.
        In particular, the orange closed connected region is not the closure of a union of activation regions.
        }
	\label{fig:zanotti_network}
\end{figure}

By dropping the `connectedness' requirement, we obtain the following definition.
\begin{definition}[Affine Region]
\label{def:region6}
    Given a ReLU network $N$, an \emph{affine region} of $N$ is a full-dimensional inclusion-wise maximal subset of $\mathbb{R}^n$ on which $f_N$ is affine.
\end{definition}

For each definition, a linear region $S\subseteq \mathbb{R}^n$ of a ReLU network $N$ can be associated with an affine function $g:\mathbb{R}^n\to \mathbb{R}$ such that $f_N(x)=g(x)$ for all $x\in S$.
The affine function $g$ is unique if $S$ is full-dimensional.
We say that the function $g$ is \emph{computed} or \emph{realized} on $S$.
If $g$ is the zero function, we call $S$ a \emph{zero region} and a \emph{nonzero region} otherwise.

\subsection{Relation to minimum covers}\label{subsec:min_covers}
For \Cref{def:region3}, the number of convex regions is defined as the minimum number of full-dimensional polyhedra in a compatible polyhedral complex, as in general there is no canonical polyhedral complex attaining this minimum and hence no canonical set of convex regions.

In the following, we show that \Cref{def:region4,def:region5,def:region6} can also be interpreted in terms of cardinality-minimal sets of subsets whose union (or the union of closures) covers $\mathbb{R}^n$ and $f_N$ restricted to each subset is affine.
We will see that, similar to \Cref{def:region3}, such cardinality-minimal sets are not necessarily unique.
However, unlike for \Cref{def:region3}, there is a natural way to define a canonical cardinality-minimal set, namely the sets given by \Cref{def:region4,def:region5,def:region6}.
More precisely, \Cref{def:region4,def:region5,def:region6} directly define canonical cardinality-minimal sets of full-dimensional and inclusion-wise maximal subsets that are either open and connected (\Cref{def:region4}), closed and connected (\Cref{def:region5}), or only closed (\Cref{def:region6}).
\begin{lemma}\label{lem:min_cover_open_regions}
Given a ReLU network $N$, the set of open connected regions of $N$ is cardinality-minimal among all sets $\mathcal{S}$ of open connected subsets such that $\bigcup_{S\in \mathcal{S}}\overline{S}=\mathbb{R}^n$ and $f_{N}$ is affine on every $S\in \mathcal{S}$.
\end{lemma}
\begin{proof}
    We first observe that the set $\mathcal{R}_{\rm open}$ of open connected regions of $N$ satisfies both assumptions of the lemma.
    For this, observe that every point $x\in \mathbb{R}^n$ lies in the closure of some full-dimensional polyhedron of a fixed polyhedral complex that is compatible with $f_N$.
    The interior of this polyhedron is contained in an open connected set on which $f_N$ is affine and is therefore contained in an open connected region.
    Thus, $\bigcup_{R\in \mathcal{R}_{\rm open}}\overline{R}=\mathbb{R}^n$ and hence $\mathcal{R}_{\rm open}$ satisfies both assumptions of the lemma.

    It remains to prove that $\mathcal{R}_{\rm open}$ is cardinality-minimal.
    Let $\mathcal{S}$ be a cardinality-minimal set satisfying both assumptions of the lemma, fix an open connected set $S\in \mathcal{S}$, and let $g$ be the unique affine function with $g(x)=f_N(x)$ for all $x\in S$.
    Let $x\in S$.
    Then, there is some $R\in \mathcal{R}_{\rm open}$ with $x\in \overline{R}$.
    Since $S$ is open, this implies $S\cap R\neq \emptyset$.
    Since both $S$ and $R$ are open, their nonempty intersection $S\cap R$ is full-dimensional and the affine functions on $S$ and $R$ are identical, that is, $g(x)=f_N(x)$ for all $x\in S\cup R$.
    By the inclusion-wise maximality of $R$, we have $S\subseteq R$.
    The region $R$ is unique, since two distinct open connected regions cannot have a nonempty open intersection.
    Conversely, every open connected region $R\in \mathcal{R}_{\rm open}$ contains at least one subset from $\mathcal{S}$.
    Thus, the map that assigns to each $S\in \mathcal{S}$ its containing open connected region $R_S\in \mathcal{R}_{\rm open}$ is surjective, which implies $|\mathcal{R}_{\rm open}|\leq |\mathcal{S}|$.
    Since $\mathcal{S}$ is cardinality-minimal by assumption, the map is a bijection, proving that $\mathcal{R}_{\rm open}$ is cardinality-minimal.
\end{proof}
Similarly, we obtain the following lemmas.
We omit the proofs, as they are elementary modifications of the proof of \Cref{lem:min_cover_open_regions}.
\begin{lemma}\label{lem:min_cover_closed_conn_regions}
Given a ReLU network $N$, the set of closed connected regions of $N$ is cardinality-minimal among all sets $\mathcal{S}$ of full-dimensional closed connected subsets such that $\bigcup_{S\in \mathcal{S}}S=\mathbb{R}^n$ and $f_{N}$ is affine on every $S\in \mathcal{S}$.
\end{lemma}
\begin{lemma}\label{lem:min_cover_affine_regions}
Given a ReLU network $N$, the set of affine regions of $N$ is cardinality-minimal among all sets $\mathcal{S}$ of full-dimensional closed subsets such that $\bigcup_{S\in \mathcal{S}}S=\mathbb{R}^n$ and $f_{N}$ is affine on every $S\in \mathcal{S}$.
\end{lemma}

There can be different cardinality-minimal sets of closed connected subsets (and only closed subsets).
For example, in \Cref{fig:activation_regions}, in such a minimal set $\mathcal{S}$ there is exactly one closed subset corresponding to the closed connected region with the constant zero function and there are multiple options to choose this subset, e.g. $(-\infty, 0] \times [0,\infty)$ or $((-\infty, 0] \times [0,\infty)) \cup \{(x,0):x>0\}$.
The same is true for open connected regions.
For example, $\{\mathbb{R}^2\}$ and $\{\mathbb{R}^2\setminus\{0\}\}$ are both cardinality-minimal sets of open connected sets satisfying the assumptions of \Cref{lem:min_cover_open_regions} for any network computing an affine function on $\mathbb{R}^2$.
Observe that $\mathbb{R}^2$ is regular open, whereas $\mathbb{R}^2\setminus\{0\}$ is not.
In fact, one can show that the open connected regions are the unique cardinality-minimal cover if one additionally requires the open connected subsets of the cover to be regular open.
\begin{lemma}\label{lem:unique_min_regular_open_cover_open_regions}
    Let $N$ be a ReLU network and let $\mathcal{R}_{\rm open}$ be the set of its open connected regions.
    Then every $R\in \mathcal{R}_{\rm open}$ is regular open.
    Moreover, $\mathcal{R}_{\rm open}$ is the unique cardinality-minimal set $\mathcal{S}$ of regular open connected subsets such that $\bigcup_{S\in \mathcal{S}}\overline{S}=\mathbb{R}^n$ and $f_{N}$ is affine on every $S\in \mathcal{S}$.
\end{lemma}
\begin{proof}
    We first show that every open connected region is regular open.
    Fix $R\in \mathcal{R}_{\rm open}$ and let $g$ be the affine function with $g(x)=f_N(x)$ for all $x\in R$.
    Clearly, this implies that $g(x)=f_N(x)$ for all $x\in (\overline{R})^\circ$.
    We claim that $(\overline{R})^\circ$ is connected.
    For this, recall that an open set $U\subseteq \RR^n$ is connected if and only if there are two nonempty open sets $U_1,U_2\subseteq \RR^n$ such that $U_1\cap U_2=\emptyset$ and $U_1\cup U_2=U$.
    Suppose now that $(\overline{R})^\circ$ is not connected, and let $U_1,U_2\subseteq \RR^n$ be nonempty disjoint open sets with $U_1\cup U_2=(\overline{R})^\circ$.
    Then, $U_1'=U_1\cap R$ and $U_2'=U_2\cap R$ are two nonempty disjoint open sets with $U_1'\cup U_2'=R$, contradicting that $R$ is connected.
    Thus $(\overline{R})^\circ$ is connected.
    Since $(\overline{R})^\circ$ is an open connected set containing $R$ on which $f_N$ is affine, the inclusion-wise maximality of $R$ implies $R=(\overline{R})^\circ$.
    
    It remains to prove the uniqueness statement.
    Let $\mathcal S$ be a cardinality-minimal set satisfying the assumptions.
    By \Cref{lem:min_cover_open_regions}, the map $R\mapsto S_R$ which assigns each $R\in \mathcal{R}_{\rm open}$ the unique set $S_R\in \mathcal{S}$ contained in $R$ (so $S_R\subseteq R$) is a bijection.
    We claim that $R\subseteq \overline{S_R}$.
    Let $x\in R$.
    Then, there is some $S\in \mathcal S$ with $x\in \overline{S}$.
    Since $R$ is open, we have $R\cap S\neq \emptyset$, and by the same argument as in \Cref{lem:min_cover_open_regions}, $S$ is contained in $R$.
    By uniqueness, $S=S_R$ and hence $x\in S_R$.
    We thus have $S_R\subseteq R\subseteq \overline{S_R}$ and therefore $\overline{S_R}=\overline R$.
    Since both sets are regular open, we have $S_R=(\overline{S_R})^\circ=(\overline{R})^\circ=R$ and thus $\mathcal S=\mathcal{R}_{\rm open}$.
\end{proof}

The following theorem is immediate from the definitions and the subsequent observations.
\begin{theorem}
\label{thm:region_hierarchy}
    Given a ReLU network $N$, let $R_1,\dots,R_6$ denote the number of activation regions, proper activation regions, convex regions, open connected regions, closed connected regions and affine regions, respectively.
    Then,
    $
    R_6 \leq R_5 \leq R_4 \leq R_3 \leq R_2 \leq R_1.
    $
\end{theorem}

The examples in \Cref{fig:activation_regions,fig:linear_regions} show that each inequality in \Cref{thm:region_hierarchy} can be strict.
\citet{zanotti2025bounds} showed $R_4\in O((R_6)^{n+1})$.
\citet{he2018relu} showed $R_3\leq (R_6)!$.
Trivially, $R_1\leq 2^{s(N)}$.

\begin{figure}
    \centering
    \hfill
    \includegraphics[page=4, width=0.3\textwidth]{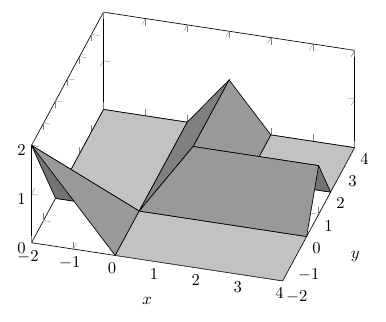}\hfill
    \includegraphics[page=3, width=0.25\textwidth]{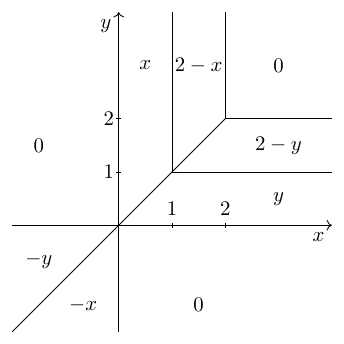}\hfill{}
    \caption{The function $\max(0,x)+\max(0,-y) - \max(0,x-y) +\min(\max(0, x-2), \max(0, y-2)) - 2 \min(\max(0, x-1), \max(0, y-1))$.
    We have $R_6=7$, $R_5=8$, and $R_4=9$.}
    \label{fig:linear_regions}
\end{figure}

\subsection{Inaccuracies in the literature}
\label{sec:flaws}
In this subsection we list a few cases where misunderstandings about the different definitions of a linear region led to small errors or inconsistencies in previous work, alongside with a suggestion how they would be fixable. Usually, this can be achieved by switching to a different, maybe more appropriate definition of a linear region.

In Lemma 11 (d), \citet{chen2022improved} claim that $\left(S_1\cap S_2\right) \cap\left(S_1^\circ \cup S_2^\circ\right) = \emptyset$ holds for any two closed connected regions $S_1, S_2$. While this claim is not true for closed connected regions, it is true that $\left(\overline{S_1}\cap \overline{S_2}\right)\cap\left(S_1^\circ \cup S_2^\circ\right) = \emptyset$ holds for any two \emph{open} connected regions $S_1, S_2$.
This inaccuracy was first pointed out by \citet{zanotti2025bounds}.
	
The algorithm of \citet{wang2022estimation} does not create the set of closed connected regions, since in the algorithm, only closures of activation regions are merged such that (1) the affine function of the regions is identical and (2) the closures of the two activation regions have a non-empty intersection.
However, a closed connected region is in general \emph{not} a union of the closure of a set of activation regions, see \Cref{fig:zanotti_network}.
However, the algorithm of \citet{wang2022estimation} can instead be adapted to count the number of \emph{open connected regions} with one minor adjustment. Instead of merging the closure of activation regions with a nonempty intersection, merge only two proper activation regions if their closure has a $(n-1)$-dimensional intersection. This computation can still be done in time polynomial in the input size through linear programming, as described in \Cref{lem:dim_of_activation_region_is_poly_time_computable}.

\citet[Lemma 3.2]{zanotti2025bounds} claims that a minimum-cardinality set of open connected pieces consists of inclusion-wise maximal pieces.
This is not true in general: for an affine function on $\RR^2$, the piece $\RR^2\setminus\{0\}$ is open, connected and has minimum cardinality, but is not maximal.
This inaccuracy can be fixed by requiring the subsets in their Definition 1.1 to be regular open and connected, see also \Cref{lem:unique_min_regular_open_cover_open_regions}.

\subsection{Problem definitions for counting linear regions}
In the remainder of the paper, we consider algorithmic problems arising from counting linear regions. Both decision problems (e.g., deciding if the number of regions is larger than a given threshold) or function problems (such as computing exactly or approximately the number of linear regions) are detailed below.

\textsc{$K$-region-decision}
\\
\textbf{Input:} A ReLU network $N$.
\\
\textbf{Question:} Does $N$ have strictly more than $K$ linear regions (according to a specified definition)?

\textsc{Linear region counting}
\\
\textbf{Input:} A ReLU network $N$.
\\
\textbf{Question:} What is the number of linear regions of $N$ (according to a specified definition)?

\section{Counting regions: one hidden layer}
\label{sec:shallow}
In this section, we derive our results for ReLU networks with one hidden layer.
Our first result is that \textsc{1-region-decision} can be solved in polynomial time for ReLU networks with one hidden layer.

Let $N$ be a ReLU network with one hidden layer and let $[n_1]$ denote the set of neurons.
For ease of notation, we denote $A^{(2)}=a\in \mathbb{R}^{n_1}$, $A^{(1)}=(w_1,\dots, w_{n_1})^\top\in \mathbb{R}^{n_1\times n}$, $b^{(2)}=B\in \mathbb{R}$, and $b^{(1)}=b\in \mathbb{R}^{n_1}$.
Then, the function computed by the network is
\[
f_N(x)=B+\sum_{i=1}^{n_1} a_i \max(0, w_i^\top x+b_i).
\]
For every neuron $i\in [n_1]$, we define the set $H_i := \{x\in \mathbb{R}^n : w_i^\top x+b_i = 0\}$.

\begin{theorem}
\label{thm:polyresultfor1layerdecision}
\textsc{1-region-decision} for ReLU networks with one hidden layer is in \P{} for \Cref{def:region1,def:region2,def:region3,def:region4,def:region5,def:region6}.
\end{theorem}
\begin{proof}
	We only consider \Cref{def:region3,def:region4,def:region5,def:region6} here (for \Cref{def:region1,def:region2}, the answer is trivially yes if and only if $A^{(1)}$ has a nonzero row).
	Let $N$ be a ReLU network as defined above.
	
	For \Cref{def:region3,def:region4,def:region5,def:region6}, $N$ has a single linear region if and only if $f_N$ is affine, that is, $f_N$ has no breakpoints.
    We ignore neurons with $w_i=0$, since they only contribute constants to $f_N$.
    Every neuron $i$ with $w_i\neq 0$ computes the function $\max(0, w_i^\top x+b_i)$ that has breakpoints exactly along the hyperplane $H_i$.
    However, it is not guaranteed that these breakpoints are also breakpoints of $f_N$, either because $a_i=0$ or because other neurons corresponding to the same hyperplane can cancel out the nonlinearity.
    The algorithmic idea is now to check whether $f_N$ has breakpoints along $H_i$.
    
    For this, define $g_j(x)=w_j^\top x+b_j$ for every neuron $j$ and partition the neurons into sets $\mathcal{I}=\{I_1,\dots, I_m\}$ according to their hyperplanes, that is, $H_i=H_j$ for all $i,j\in I_k$ and $k\in [m]$.
    We can compute this partition in polynomial time, since $H_i=H_j$ holds if and only if there is a constant $c_{ij}\neq 0$ such that $(w_i,b_i)=c_{ij}(w_j,b_j)$.

    Now, fix an element $I\in \mathcal{I}$ and an element $r\in I$.
    For each $i\in I$, let $c_i\neq 0$ satisfy $g_i=c_i g_r$ and split $I$ into the two sets $I^+=\{i\in I: c_i>0\}$ and $I^-=\{i\in I: c_i<0\}$.
	The sum of the functions of the neurons in $I$ is
	\begin{align*}
    \sum_{i\in I}a_i\max(0,c_i g_r(x))
    &=
    \sum_{i\in I^+}a_ic_i\max(0,g_r(x))-\sum_{i\in I^-}a_ic_i\max(0,-g_r(x))\\
    &=
    g_r(x)\sum_{i\in I^-}a_ic_i+\max(0,g_r(x))\left(\sum_{i\in I^+}a_ic_i-\sum_{i\in I^-}a_ic_i\right).
	\end{align*}
	Thus, the sum of the functions of all neurons with breakpoints along $H_r$ is affine if and only if
	\begin{equation}\label{eq:hyperplane_removal_condition}
		\sum_{i\in I^+}a_ic_i=\sum_{i\in I^-}a_ic_i.
	\end{equation}
    If \eqref{eq:hyperplane_removal_condition} does not hold for an element $I\in \mathcal{I}$, then $f_N$ is not affine.
    Conversely, if \eqref{eq:hyperplane_removal_condition} holds for all elements $I\in \mathcal{I}$, then $f_N$ is affine.
    Thus, checking \eqref{eq:hyperplane_removal_condition} for all elements $I\in \mathcal{I}$ decides whether $f_N$ is affine (and thus \textsc{1-region-decision}) in polynomial time.
\end{proof}

\citet[Lemma 15]{froese2024complexity} give a result similar to \Cref{thm:polyresultfor1layerdecision}.
They show that for a ReLU network with one hidden layer without biases, one can determine in polynomial time whether the network computes the constant zero function, and otherwise find a point on which the network computes a nonzero value.
In contrast, \Cref{thm:polyresultfor1layerdecision} considers biases and nonzero affine functions.

Turning to the problem of exactly counting the number of regions, we show the following theorem.
\begin{theorem}
\label{thm:linear_region_counting_is_sharp_p_complete_for_one_hidden_layer}
    \textsc{Linear region counting} for ReLU networks with one hidden layer is \#\P-hard for \Cref{def:region5,def:region6} and \#\P-complete for \Cref{def:region1,def:region2,def:region3,def:region4}.
\end{theorem}
\begin{proof}
	We first prove \#\P-hardness by reducing from counting the number of cells in a hyperplane arrangement, which is \#\P-complete (see~\citet{linial1986hard}).
	Let $\mathcal{H}=(H_i)_{i\in [m]}$ be a hyperplane arrangement in $\mathbb{R}^n$ with $H_i:=\{x: w_i^\top x=0\}$, $w_i\in \mathbb{Q}^n\setminus \{0\}$ such that the hyperplanes are pairwise distinct.
	Restricted to such hyperplane arrangements, the problem of counting the number of cells remains \#\P-complete, see~\citet{linial1986hard}.
    We first compute a point $x^*\in \mathbb{R}^n$ that lies in the interior of a cell in polynomial time (e.g. by successively solving linear programs).
    Given such a point, we then orient every hyperplane such that $w_i^\top x^* > 0$ for all $i\in [m]$ and construct the one-hidden layer ReLU network $N_\mathcal{H}$ computing the convex function
	\[
	f_{N_\mathcal{H}}(x)=\sum_{i=1}^m \max(0, w_i^\top x).
	\]
	We now show that, for each of the \Cref{def:region1,def:region2,def:region3,def:region4,def:region5,def:region6}, $N_\mathcal{H}$ has exactly as many linear regions as the hyperplane arrangement $\mathcal{H}$ has cells.
    First, let $S_a$ be a nonempty activation region with support $I\subseteq [m]$, and choose an $x\in S_a$.
    Since $w_i^\top x>0$ for all $i\in I$ and $w_i^\top x\leq 0$ for all $i\notin I$, and $w_i^\top x^*>0$ for all $i\in [m]$, there is a sufficiently small $\lambda>0$ such that $x_\lambda=(1-\lambda)x-\lambda x^*$ satisfies
    \[
    w_i^\top x_\lambda>0 \text{ for } i\in I,
    \qquad
    w_i^\top x_\lambda<0 \text{ for } i\notin I.
    \]
    Hence $x_\lambda$ is an interior point of $S_a$ and thus every activation region is full-dimensional.
    Therefore the activation regions, proper activation regions, and cells of $\mathcal{H}$ are in one-to-one correspondence.
	
	We now show that no two cells have the same affine function.
	Let $I,I'\subseteq [m]$ denote the supports of the activation patterns of two distinct cells and suppose that the affine functions on those cells were equal.
	Since the affine function of $f_{N_\mathcal{H}}$ on an activation region with support $I\subseteq[m]$ is $x\mapsto x^\top\left(\sum_{i\in I}w_i\right)$, we have
	\[
	\sum_{i\in I}w_i = \sum_{i\in I'}w_i
	\]
	and thus $\sum_{i\in I\setminus I'}w_i = \sum_{i\in I'\setminus I}w_i$.
	Let $y$ be an interior point of the cell with support $I$.
	Then, we have $w_i^\top y>0$ for every $i\in I\setminus I'$ and $w_i^\top y<0$ for every $i\in I'\setminus I$, which yields
	\[
	y^\top\left(\sum_{i\in I\setminus I'}w_i\right)  > 0>y^\top\left(\sum_{i\in I'\setminus I}w_i\right),
	\]
	contradicting  $\sum_{i\in I\setminus I'}w_i = \sum_{i\in I'\setminus I}w_i$.
	Hence, no two cells share the same affine function.
    Therefore, the number of affine regions is equal to the number of cells, which proves that the number of linear regions is the same for all six definitions, proving \#\P-hardness for all six definitions.
	
	We now show that \textsc{Linear region counting} is in \#\P\ for \Cref{def:region1,def:region2,def:region3,def:region4}.
	A certificate for \Cref{def:region1,def:region2} is simply an activation pattern $a$ of the ReLU network $N$, which can be checked in polynomial time by computing the dimension of $S_a$, see \Cref{lem:dim_of_activation_region_is_poly_time_computable}.
	Since each activation pattern is a unique certificate for a linear region under these definitions, this gives containment in \#\P\ for \Cref{def:region1,def:region2}.
    
	For \Cref{def:region3,def:region4}, we first compute the set $\mathcal{H}=(H_i)_{i\in[m]}$ of hyperplanes that correspond to breakpoints of $f_N$ using the algorithm from the proof of \Cref{thm:polyresultfor1layerdecision}.
	By construction, $f_N$ is affine on each cell of $\mathcal{H}$, and the affine functions that are realized on two neighboring cells (cells with an $(n-1)$-dimensional intersection) cannot be equal.
	Thus, the cells of $\mathcal{H}$ are precisely the open connected regions of $N$.
	Since every open connected region is a cell and thus convex, the number of convex regions is equal to the number of open connected regions.
    A cell $C$ is uniquely defined by its sign vector $a\in\{-,+\}^m$, where $a_i$ indicates on which side of $H_i$ the cell $C$ lies.
    Thus, the sign vector with respect to the hyperplane arrangement $\mathcal{H}$ gives a unique certificate for an open connected or a convex region, which can be checked in polynomial time using the algorithm from \Cref{thm:polyresultfor1layerdecision} and linear programming.
	It follows that the problem is in \#\P{} for \Cref{def:region3,def:region4}.
\end{proof}
\citet{linial1986hard} proved the \#\P-completeness of counting the number of cells of a hyperplane arrangement by reducing from the \#\P-complete problem of counting the number of acyclic orientations of a graph.
Their reduction implies that \textsc{Linear region counting} is \#\P-hard even for one-hidden-layer networks where $A^{(2)}=(1,\dots,1)$, $b^{(1)}=0$, $b^{(2)}=0$, and $A^{(1)}$ is the transpose of an incidence matrix of a directed graph.

It is an open problem whether \textsc{Linear region counting} is in \#\P{} for \Cref{def:region5,def:region6}.
Notice that a single activation pattern does not suffice as a certificate, since two proper activation regions whose closures have a non-empty intersection can have the same affine function.
For example, consider the function $\max(0,x)+\max(0,-y)-\max(0,x-y)$ with zero regions $(-\infty, 0]\times [0,\infty)$ and $[0,\infty) \times (-\infty, 0]$.

To the best of our knowledge, it is unknown whether there is a polynomial factor approximation algorithm for approximating the number of cells in a hyperplane arrangement.
Thus, it is also an open problem whether \textsc{Linear region counting} for ReLU networks with one hidden layer has a polynomial factor approximation algorithm that runs in polynomial time.

\section{Counting regions: going beyond one hidden layer}
\label{sec:deep}
In this section, we prove our hardness results for ReLU networks with more than one hidden layer for \Cref{def:region3,def:region4,def:region5,def:region6}.
We start by proving NP-completeness of \textsc{1-region-decision} for networks with two hidden layers for \Cref{def:region3,def:region4,def:region5,def:region6} in \Cref{thm:np_completeness_2region_decision_problem}.
Note that the same problem is solvable in polynomial time for \Cref{def:region1,def:region2}, since the answer is trivially yes if and only if the weight matrix of the first layer has a nonzero row.
Since the other hardness results of this section are all based on \Cref{thm:np_completeness_2region_decision_problem}, we obtain hardness for these other problems for \Cref{def:region3,def:region4,def:region5,def:region6}, but not for \Cref{def:region1,def:region2}.

\subsection{Hardness of the decision version}
The following theorem follows immediately from a result of \citet{wang2022estimation}.
\begin{theorem}[\citet{wang2022estimation}]
    For any fixed constant $K\in \mathbb{N}_{\geq 1}$, \textsc{$K$-region-decision} for ReLU networks of depth $\Theta(\log n)$ is \NP-hard according to \Cref{def:region3,def:region4,def:region5,def:region6}.
\end{theorem}
In their reduction from 3-\SAT, they construct a network computing a minimum of $n+1$ terms.
The currently best known constructions for computing the minimum (or, equivalently, the maximum) of $n+1$ terms require depth $\Theta(\log n)$ \cite{ruess2026shallower}, this leads to hardness for counting regions of networks with depth $\Theta(\log n)$.
We improve on the result by showing that the problem remains $\NP$-hard even for networks with two hidden layers.
As a first step, we show that \textsc{1-region-decision} for ReLU networks with two hidden layers is \NP-complete according to \Cref{def:region3,def:region4,def:region5,def:region6}.

\begin{figure}
    \centering
    \includegraphics[page=5, width=0.3\textwidth]{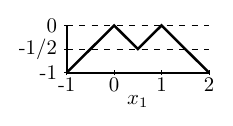}
    \includegraphics[page=6, width=0.325\textwidth]{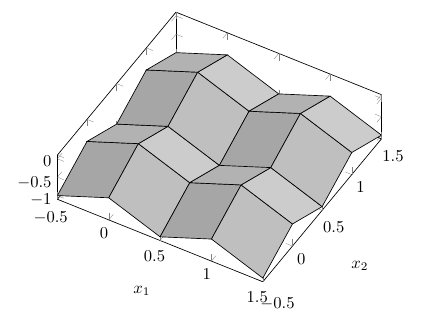}
    \includegraphics[page=7, width=0.275\textwidth]{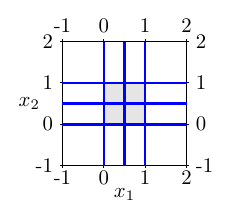}
    \caption{The function $T_1$ (left), $T_2$ (center) and its linear regions (right).}
    \label{fig:subtraction_function3}
\end{figure}
\textbf{Proof idea.}
We first give a high-level overview of the intuition behind the proof.
Similar to \citet{wang2022estimation}, we use the fact that one can simulate AND, OR, and NOT gates via ReLU networks on binary inputs which was already observed by \citet[Appendix I]{katz2017reluplex}.
More specifically, given a \SAT{} formula $\phi(x)=\bigwedge_{i=1}^m((\bigvee_{j\in P_i}x_j) \vee (\bigvee_{j\in N_i}\neg x_j))$ on the variables $x_1,\dots, x_n$, the function
\[
g_\phi(x) = 1 - \sum_{i=1}^m \max\left(0, 1-\sum_{j\in P_i}x_j-\sum_{j\in N_i}(1-x_j)\right)
\]
takes value $1$ on all $\{0,1\}^n$ vectors satisfying $\phi$ and value at most $0$ on all non-satisfying binary vectors, since the $i$-th summand in the above sum evaluates to $0$ for all binary vectors that satisfy the $i$-th clause of $\phi$, and to $1$ otherwise.
Thus, for every $\varepsilon\in (0,1)$, the function
\[
h_{\phi,\varepsilon}=\max(0,g_\phi-1+\varepsilon)
\]
has at least two regions according to \Cref{def:region3,def:region4,def:region5,def:region6} if $\varphi$ is satisfiable, since $h_{\phi,\varepsilon}$ evaluates to $\varepsilon$ on a satisfying binary vector, and to $0$ on all points in an open set around a non-satisfying binary vector (see \Cref{fig:sat_example_figure} for an example).
Note that we can assume that $\phi$ has a non-satisfying binary vector since each nonempty clause in $\phi$ is not satisfied by at least one binary vector.
Thus, if $h_{\phi,\varepsilon}$ had only one linear region when $\varphi$ is unsatisfiable, we would have a reduction from \SAT{} to \textsc{1-region-decision} for two-hidden-layer ReLU networks.
However, this is not true in general.
For example, the \SAT{} formula $\psi:\{0,1\}^2\to \{0,1\},\; (x_1,x_2)\mapsto(x_1\lor x_2)\land (\neg x_1\lor x_2)\land (x_1\lor \neg x_2)\land (\neg x_1\lor \neg x_2)$ is unsatisfiable, but $h_{\psi, \varepsilon}(1/2,1/2)=\varepsilon>0$ holds for every $\varepsilon\in (0,1)$.

The key idea to resolve this is to add a function that is negative everywhere except at the elements of $\{0,1\}^n$ (on which it evaluates to zero).
A function with this property is $T_n:\mathbb{R}^n\to \mathbb{R},\;x\mapsto\sum_{i=1}^n\tau(x_i)$ (shown in \Cref{fig:subtraction_function3}) with $\tau(x)
=-\max(0,-x)-\max(0,x)+\max(0,2x-1)-\max(0,2x-2)$.
$\tau(x)$ is equal to $x$ for $x\leq 0$, $-x$ for $0\leq x \leq 1/2$, $x-1$ for $1/2 \leq x \leq 1$, and $1-x$ for $x\geq 1$ (also see \Cref{fig:subtraction_function3}).
Thus $T_n(x)\leq 0$ for all $x\in \RR^n$, with equality exactly at $x\in \{0,1\}^n$.
Moreover, for $0<\varepsilon<1/2$, if some coordinate lies outside $(-\varepsilon,\varepsilon)\cup (1-\varepsilon,1+\varepsilon)$, then $T_n(x)\leq -\varepsilon$.
We now show that adding this function recovers the equivalence: A \SAT{} formula is satisfiable if and only if the function $\max(0, T_n + \varepsilon-1 + g_\phi)$ has more than one linear region according to \Cref{def:region3,def:region4,def:region5,def:region6} (see \Cref{fig:sat_example_figure} for an example).

\begin{figure}[b]
    \centering
    \includegraphics[page=11, width=0.25\textwidth]{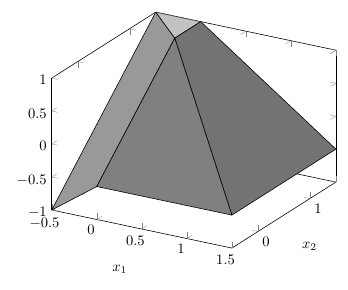}
    \includegraphics[page=13, width=0.25\textwidth]{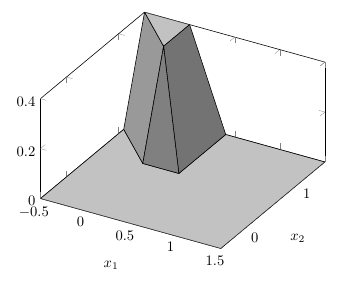}
    \includegraphics[page=15, width=0.25\textwidth]{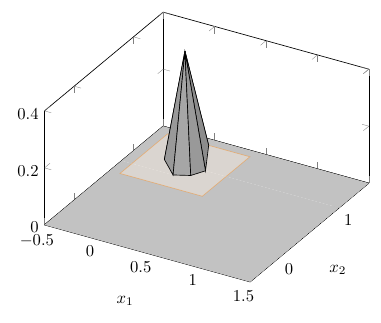}
    \includegraphics[page=16, width=0.22\textwidth]{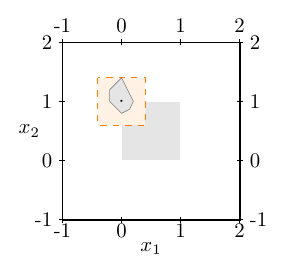}
    \caption{
    The functions $g_\phi(x)=1-\max(0,x_1)-\max(0,1-x_1-x_2)$ (left), $h_{\phi, \varepsilon}(x)=\max(0, \varepsilon-1+g_\phi(x))$ (center left), and $f_{N_\phi}(x)=\max(0,T_2(x)+\varepsilon-1+g_\phi(x))$ (right) for the \SAT\ formula $\phi(x) = (\neg x_1)\land (x_1 \lor x_2)$ with the satisfying assignment $(0,1)$ and  $\varepsilon=0.4$.
    The function $f_{N_\phi}$ is only nonzero in the gray region, which is contained in the $\varepsilon$-square (orange) around the only satisfying point of $\phi$ (black).
    }
    \label{fig:sat_example_figure}
\end{figure}

\begin{lemma}\label{thm:np_completeness_2region_decision_problem}
\textsc{1-region-decision} for two-hidden-layer ReLU networks is \NP-complete for \Cref{def:region3,def:region4,def:region5,def:region6}.
\end{lemma}
\begin{proof}
We first show that the problem is in \NP.
For each of \Cref{def:region3,def:region4,def:region5,def:region6}, a network is a yes-instance of \textsc{1-region-decision} if and only if its computed function is not affine.
If the function is not affine, then two proper activation regions realize two distinct affine functions.
Their activation patterns serve as a polynomial certificate of a yes-instance.
By \Cref{lem:pattern_to_function_is_poly_time_computable,lem:dim_of_activation_region_is_poly_time_computable}, we can check in polynomial time if both activation regions are proper and that their affine functions are distinct.
Thus, the problem is in \NP.

For \NP-hardness, we reduce from \SAT.
Let $\phi(x)=\bigwedge_{i=1}^m((\bigvee_{j\in P_i}x_j) \vee (\bigvee_{j\in N_i}\neg x_j))$ be a \SAT\ formula on $n$ variables.
Since duplicate literals and tautological clauses can be removed in polynomial time, we can assume $P_i\cap N_i=\emptyset$ and $|P_i|+|N_i|\leq n$ for all $i\in [m]$ without loss of generality.
Set $\varepsilon = 1/(n+1)$.
Consider the network $N_\phi$ computing the function
\[
    f_{N_\phi}(x)
    =\max\left(0, T_n(x)+\varepsilon-\sum_{i=1}^m \max\left(0, 1-\sum_{j\in P_i}x_j-\sum_{j\in N_i}(1-x_j)\right)\right).
\]
The inner ReLU terms and the ReLU terms of $T_n$ form the first hidden layer, and the outer maximum is computed in the second hidden layer.
Thus $N_\phi$ has two hidden layers and polynomial encoding size.

As discussed above, we have $T_n(x)\leq -\varepsilon$ and thus $f_{N_\phi}(x)=0$ whenever some $x_i$ lies outside $(-\varepsilon,\varepsilon)\cup (1-\varepsilon,1+\varepsilon)$.
Therefore
\[
\left\{x\in \RR^n: f_{N_\phi}(x)>0\right\}\subseteq \bigcup_{z\in \{0,1\}^n}B^\infty_\varepsilon(z),
\qquad
B^\infty_\varepsilon(z):=\{x: \|x-z\|_\infty \leq \varepsilon\}.
\]

If $z\in \{0,1\}^n$ satisfies $\phi$, then every clause contains a true literal and thus $\sum_{j\in P_i}z_j+\sum_{j\in N_i}(1-z_j)\geq 1$ for all $i\in [m]$.
Therefore
\[
f_{N_\phi}(z)=\max(0, T_n(z)+\varepsilon)=\varepsilon>0.
\]
Since $f_{N_\phi}$ is zero outside of the union of the $B_\varepsilon^\infty(z)$-boxes, it is not affine and has more than one region.

Suppose $z\in \{0,1\}^n$ does not satisfy $\phi$.
There is at least one clause $i^*$ with $\sum_{j\in P_{i^*}}z_j+\sum_{j\in N_{i^*}}(1-z_j)= 0$.
Then, for all $x\in B^\infty_\varepsilon(z)$, we have
\[
1-\sum_{j\in P_{i^*}}x_j-\sum_{j\in N_{i^*}}(1-x_j)
\geq
1 - |P_{i^*}|\varepsilon - |N_{i^*}|\varepsilon
\geq 1 - n \varepsilon
=\varepsilon.
\]
Therefore, we have $f_{N_\phi}(x)=0$ for all $x\in B^\infty_\varepsilon(z)$.
If $\phi$ is unsatisfiable, this holds for every $z\in \{0,1\}^n$ and it follows that $f_{N_\phi}\equiv 0$.
Thus $N_\phi$ has more than one linear region if and only if $\phi$ is satisfiable.
\end{proof}

We note that there is an alternative way to prove the \NP-hardness of \Cref{thm:np_completeness_2region_decision_problem}.
\citet[Theorem 18]{froese2024complexity} show that the problem of deciding whether or not a network without biases with one hidden layer has a point which evaluates to a positive value is \NP-complete.
By taking the maximum of the output of the network used in their reduction with the zero function, we obtain the \NP-hardness of \Cref{thm:np_completeness_2region_decision_problem}.
However, we build upon the ideas used in our reduction in \Cref{sec:hardness of approximation} to obtain results on the hardness of approximation of \textsc{Linear region counting}.
Moreover, our reduction has different properties, for example, all nonzero linear regions are bounded.
This is not possible without biases, since then, all nonzero regions correspond to a union of polyhedral cones.
The following theorem extends the hardness result of \Cref{thm:np_completeness_2region_decision_problem}.

\begin{theorem}\label{thm:np_hardness_Kregion_decision}
    For any fixed constants $K, L\in \mathbb{N}_{\geq 1}$, $L\geq 2$, \textsc{$K$-region-decision} for ReLU networks with $L$ hidden layers is \NP-hard for \Cref{def:region3,def:region4,def:region5,def:region6}.
\end{theorem}
\begin{proof}
We reduce from \SAT.
Fix $K, L\in \mathbb{N}_{\geq 1}$, $L\geq 2$.
For $K=1$, take the network $N_\phi$ from \Cref{thm:np_completeness_2region_decision_problem}.
For $K\geq 2$, define
\[
q_K:\RR^n \to \RR,\qquad x\mapsto -\sum_{r=2}^K\max(0,2(n+m)(x_1-r))
\]
and let the network $N^{(K)}_\phi$ compute $f_{N_\phi}+q_K$.
The function $q_K$ has exactly $K$ regions according to \Cref{def:region3,def:region4,def:region5,def:region6}, with breakpoints at $x_1=2,\dots, x_1=K$.
All nonzero linear regions of $f_{N_\phi}$ lie inside the hypercube $[-\varepsilon,1+\varepsilon]^n$, so the breakpoints of $q_K$ are disjoint from those of $f_{N_\phi}$.

If $\phi$ is unsatisfiable, then $f_{N_\phi}\equiv 0$, and $N^{(K)}_\phi$ has exactly $K$ regions according to \Cref{def:region3,def:region4,def:region5,def:region6}.
If $\phi$ is satisfiable, then $f_{N_\phi}$ produces at least one nonzero region inside $[-\varepsilon,1+\varepsilon]^n$.
Observe that the coefficient of $x_1$ of every nonzero affine region induced by $q_K$ is at most $-2(n+m)$, while the $x_1$-coefficient of every affine function of $f_{N_\phi}$ is at least $-(n+m)$, which can be easily seen from the proof of \Cref{thm:np_completeness_2region_decision_problem}.
Hence, the affine function of the additional nonzero region of $f_{N_\phi}$ is not already produced by $q_K$, so $N^{(K)}_\phi$ has more than $K$ affine regions (\Cref{def:region6}) and therefore also more than $K$ regions according to \Cref{def:region3,def:region4,def:region5}.

The additional maximum terms of $q_K$ can be created by using two neurons in the first hidden layer that correspond to the positive and negative part of $x_1$, respectively, and adding $K-1$ neurons in the second hidden layer (using the equation $x_1=\max(0, x_1)-\max(0,-x_1)$ to build the maximum terms in the second hidden layer).
Note that, since $K$ is fixed, this increases the network size only polynomially.

If $L=2$, we take the network $N^{(K)}_\phi$.
If $L\geq 3$, we add $L-2$ new hidden layers between the output layer and the last hidden layer that compute the identity.
Each new hidden layer has two neurons, the first neuron outputs $\max(0, N^{(K)}_\phi)$ and the second neuron outputs $\max(0, -N^{(K)}_\phi)$.
We achieve this by replacing the arcs from the second hidden layer of $N^{(K)}_\phi$ to the output node by connections to the two neurons of the first newly added hidden layer.
The resulting network has $L$ hidden layers and polynomial encoding size (since $L$ is fixed), completing the reduction from \SAT.
\end{proof}

As a corollary of \Cref{thm:np_hardness_Kregion_decision}, we obtain insights on the following decision problem.

\textsc{$L$-network-equivalence}
\\
\textbf{Input:} Two ReLU networks $N, N'$ with $L$ hidden layers and the same input dimension.
\\
\textbf{Question:} Do the networks $N$ and $N'$ compute the same function?

Two ReLU networks compute the same function if and only if the difference of the networks is the zero function.
Since this difference can be computed by a single ReLU network, we obtain the following.
\begin{corollary}
\label{cor:euqality_of_two_relu_networks}
    \textsc{1-network-equivalence} is in \P{}, and, for any fixed constant $L\geq 2$, \textsc{$L$-network-equivalence} is $\coNP$-complete.
\end{corollary}
\begin{proof}
    Given two ReLU networks $N, N'$ with $L$ hidden layers, we can compute their difference in parallel with a network $N^-$ with $L$ hidden layers.
    Then $N$ and $N'$ compute the same function if and only if $f_{N^-}\equiv 0$.

    As in the proof of \Cref{thm:np_completeness_2region_decision_problem}, an activation pattern of a proper activation region of $N^-$ on which a nonzero affine function is computed can be used as a certificate for a no-instance, showing that \textsc{$L$-network-equivalence} is in \co\NP.

    For $L=1$, \Cref{thm:polyresultfor1layerdecision} decides in polynomial time whether $f_{N^-}$ is affine.
    If it is not affine, then $f_{N^-}$ is not the zero function.
    If it is affine, then $f_{N^-}\equiv 0$ if and only if it evaluates to zero on $n+1$ affinely independent points, which yields a polynomial time algorithm for \textsc{1-network-equivalence}.
    
    For every fixed $L\geq 2$, let $N^{(1)}_\phi$ be the ReLU network with $L$ hidden layers from the proof of \Cref{thm:np_hardness_Kregion_decision}, and let $N_0$ be a ReLU network with $L$ hidden layers that computes the zero function.
    By \Cref{thm:np_hardness_Kregion_decision}, a \SAT{} formula $\phi$ is satisfiable if and only if $N^{(1)}_\phi$ and $N_0$ are a no-instance of \textsc{$L$-network-equivalence}, proving that \textsc{$L$-network-equivalence} is \coNP-hard.
\end{proof}
We also obtain the following runtime lower bound based on the Exponential Time Hypothesis.\footnote{The Exponential Time Hypothesis~\cite{impagliazzo2001complexity} states that 3-\SAT\ on $n$ variables cannot be solved in $2^{o(n)}$ time.}
\begin{corollary}
\label{cor:eth_lower_bound}
    For any fixed constants $K,L\in \mathbb{N}_{\geq 1}, L\geq 2$, \textsc{$K$-region-decision} and \textsc{Linear region counting} for \Cref{def:region3,def:region4,def:region5,def:region6} for ReLU networks with input dimension $n$ and $L$ hidden layers cannot be solved in $2^{o(n)}$ or $2^{o(\sqrt{\langle N\rangle})}$ time unless the Exponential Time Hypothesis fails.
\end{corollary}
\begin{proof}
Fix $K, L\in \mathbb{N}_{\geq 1}, L\geq 2$.
Given a 3-\SAT\ formula $\phi$ on $n$ variables and $m$ clauses, let $N^{(K,L)}_\phi$ be the ReLU network with $L$ hidden layers and input dimension $n$ from the proof of \Cref{thm:np_hardness_Kregion_decision}.
Recall that $N_\phi$ has input dimension $n$, $4n+m$ neurons in the first hidden layer, a single neuron in the second hidden layer, all weights lie in $\{0,\pm 1, \pm 2\}$, and all biases have encoding size $\mathcal{O}(\log(n))$.
Thus $\langle N_\phi\rangle\in\mathcal{O}(n^2+mn)$.
Since $K$ and $L$ are fixed, the encoding size of the network is $\langle N^{(K,L)}_\phi\rangle\in\mathcal{O}(n^2+nm)$.

Algorithms for \textsc{$K$-region-decision} running in time $2^{o(n)}$ or $2^{o(\sqrt{\langle N\rangle})}$ would imply an $2^{o(n+m)}$ algorithm for 3-SAT, contradicting the Exponential Time Hypothesis \cite{impagliazzo2001complexity}.
Since algorithms for the counting version imply algorithms for the decision version with the same running time, the same running times are also excluded for \textsc{Linear Region Counting} under the Exponential Time Hypothesis.
\end{proof}

The $2^{o(n)}$ lower bound can be seen as another example of the curse of dimensionality in machine learning.
As the input dimension grows, the problem quickly becomes intractable.

\subsection{Hardness of exact and approximate counting}
\label{sec:hardness of approximation}
In this subsection, we show that even \emph{approximating} the number of linear regions is hard for \Cref{def:region3,def:region4,def:region5,def:region6}.
We prove two inapproximability results for different network architectures.
For the first result, we use the proof ideas of \Cref{thm:np_completeness_2region_decision_problem} to prove NP-hardness of approximating the number of linear regions of ReLU networks with three or more hidden layers within an approximation ratio larger than $(2^n+1)^{-1}$.
Again, we start with describing the proof idea.

\textbf{Proof idea.}
Given a \SAT\ formula $\phi$, the network $N_\phi$ from  \Cref{thm:np_completeness_2region_decision_problem} has some nonzero linear regions near every satisfying assignment of $\phi$.
Unfortunately, the number of linear regions created per satisfying point depends on the formula $\phi$ and is not easily computable.
Therefore, we modify the network $N_\phi$ such that the same number of nonzero linear regions is created by every satisfying assignment of $\phi$.
For this, we take, for a sufficiently small $\delta>0$, the minimum (which requires one additional hidden layer) of $f_{N_\phi}$ with
\[
T_{n,\delta}:\mathbb{R}^n\to \mathbb{R},\qquad x\mapsto \max(0, \delta+T_n(x)),
\]
shown in \Cref{fig:subtraction_function4}.
We then show that $T_{n,\delta}$ is strictly smaller than $f_{N_\phi}$, but greater than zero near a satisfying binary vector.
In this way, the minimum of $f_{N_\phi}$ and $T_{n,\delta}$ is attained by $f_{N_\phi}$ near each non-satisfying binary vector (where $f_{N_\phi}$ is equal to zero) and by $T_{n,\delta}$ near each satisfying point.
We proceed by showing that exactly $2^n$ nonzero regions are created for every satisfying binary vector of $\phi$, and if $\phi$ has exactly $k$ satisfying binary vectors, the modified network has exactly $1+k\cdot 2^n$ linear regions according to \Cref{def:region4}.

\begin{figure}
    \centering
    \includegraphics[page=8, width=0.22\textwidth]{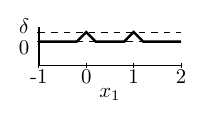}
    \includegraphics[page=9, width=0.26\textwidth]{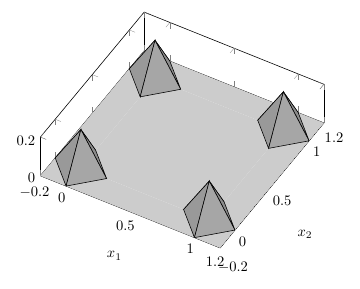}
    \includegraphics[page=10, width=0.24\textwidth]{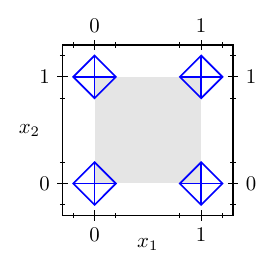}
    \includegraphics[page=19, width=0.24\textwidth]{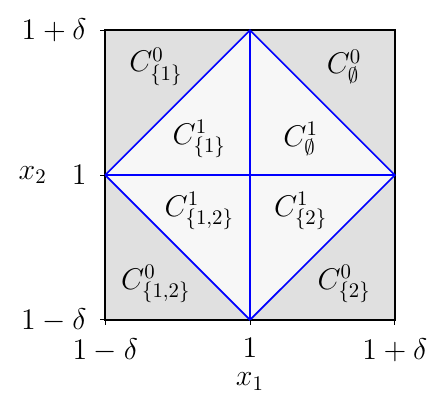}
    \caption{$T_{1,\delta}$ (left), $T_{2,\delta}$ and its linear regions (center) for $\delta=0.2$, and $C_I^0,C_I^1$ for all $I\subseteq [2]$ (right).}
    \label{fig:subtraction_function4}
\end{figure}

We start by analyzing the number of regions of $T_{n,\delta}$.
\begin{lemma}
\label{lem:teps_regions}
Let $n\geq 2$ and $B^\infty_\delta(z):=\{x\in \mathbb{R}^n: \|x-z\|_\infty \leq \delta\}$.
The function $T_{n,\delta}$ with $0<\delta< 1/2$ has exactly $1+2^{2n}$ linear regions according to \Cref{def:region4} and for every $z\in \{0,1\}^n$, the set $B^\infty_\delta(z)$ contains exactly $2^n$ nonzero regions according to \Cref{def:region3,def:region4,def:region5,def:region6}.
Further, $\{x\in \RR^n: T_{n,\delta}(x)>0\}\subseteq \bigcup_{z\in \{0,1\}^n}B^\infty_\delta(z)$ and $T_{n,\delta}(x)\leq \delta$ for all $x\in \RR^n$, with equality exactly at $x\in \{0,1\}^n$.
\end{lemma}
\begin{proof}
The last two statements follow directly from the definition of $T_n$.
We now show that $T_{n,\delta}$ has exactly $2^n$ nonzero linear regions in the set $B^\infty_\delta(z)$ for every $z\in \{0,1\}^n$.
Due to the symmetry of $T_{n,\delta}$, it is sufficient to prove this for the case where $z$ is the all-one vector, that is, $B^\infty_\delta(z)=[1-\delta,1+\delta]^n$.

First, observe that $T_n(x)= \sum_{i: x_i< 1}(x_i-1)+\sum_{i: x_i>1}(1-x_i)$ holds for any $x\in [1-\delta,1+\delta]^n$.
Given a subset $I\subseteq [n]$, we define the hyperrectangle
\[
C_I:=\left\{x: 1-\delta\leq x_i< 1 \text{ for all } i\in I,\;
1\leq x_i\leq 1+\delta \text{ for all } i\notin I\right\}.
\]
It is easy to see that the disjoint union $\bigcup_{I\subseteq [n]}C_I$ is exactly the set $[1-\delta,1+\delta]^n$.
The hyperplane $\sum_{i\in I}(x_i-1)+\sum_{i\notin I}(1-x_i)=-\delta$ cuts $C_I$ into the two sets (see \Cref{fig:subtraction_function4} for an illustration of the case $n=2$)
\begin{equation*}
    C^1_I:=\left\{x\in C_I: \sum_{i\in I}(x_i-1)+\sum_{i\notin I}(1-x_i)\geq-\delta\right\},
    \;\;
    C^0_I:=\left\{x\in C_I: \sum_{i\in I}(x_i-1)+\sum_{i\notin I}(1-x_i)\leq-\delta\right\}.
\end{equation*}
We have $T_{n,\delta}(x)=\delta+\sum_{i\in I}(x_i-1)+\sum_{i\notin I}(1-x_i)$ for all $x\in C^1_I$ and $T_{n,\delta}(x)=0$ for all $x\in C^0_I$.
The set $C^1_I$ is full dimensional, as $x^*\in \mathbb{R}^n$ with
$x^*_i=\begin{cases}
1-\frac{\delta}{2n},&i\in I\\
1+\frac{\delta}{2n},&i\notin I
\end{cases}$ is an interior point of $C^1_I$.
This proves that every $C_I$ contains exactly one nonzero region.
Since the function on $C_I$ is unique for every $I\subseteq [n]$, $[1-\delta,1+\delta]^n$ contains exactly $2^n$ nonzero regions according to \Cref{def:region3,def:region4,def:region5,def:region6}.
Since $T_{n,\delta}$ has only a single zero region according to \Cref{def:region4}, $B^\infty_\delta(z)\cap B^\infty_\delta(z')=\emptyset$ for all $z\neq z'\in \{0,1\}^n$ (since $\delta<1/2$), and every binary vector gives $2^n$ open connected regions, it follows that $T_{n,\delta}$ has exactly $1+2^{2n}$ linear regions according to \Cref{def:region4}.
\end{proof}
We now show that if $\phi$ has exactly $k$ satisfying binary vectors, the network described above has exactly $1+k\cdot 2^n$ linear regions according to \Cref{def:region4}.
In particular, this yields a reduction from $\#$\SAT.
\begin{lemma}
\label{thm:sharp_sat_reduction}
For any fixed constant $L\in \mathbb{N}, L\geq 3$, there is a reduction from $\#$\SAT\ to \textsc{Linear region counting} for networks with $L$ hidden layers according to \Cref{def:region4}.
\end{lemma}
\begin{proof}
Let $\phi(x)=\bigwedge_{i=1}^m((\bigvee_{j\in P_i}x_j) \vee (\bigvee_{j\in N_i}\neg x_j))$ be a \SAT\ formula on $n$ variables with $P_i\cap N_i=\emptyset$ and $|P_i|+|N_i|\leq n$ for all $i\in [m]$ and let $N_\phi$ be the two-hidden-layer network from the proof of \Cref{thm:np_completeness_2region_decision_problem}.
Recall that $\varepsilon=1/(n+1)$ and set $\delta = \varepsilon/(nm+1)$.
Consider the network $N_\phi^*$ that computes the function
\[
    f_{N_\phi^*}(x)=\min\left(T_{n,\delta}(x), f_{N_\phi}(x)\right),
\]
which can be computed with 3 hidden layers, since $T_{n,\delta}$ and $f_{N_\phi}$ can be computed in parallel in two hidden layers and the minimum of two terms can be expressed using three neurons $\min(a,b)=-\max(0,b-a)+\max(0,b)-\max(0,-b)$.
The reduction is polynomial since the addition of $T_{n,\delta}$ increases the encoding size only by an additional $\mathcal{O}(n^2)$ term.
Suppose that $\phi$ has exactly $k$ satisfying assignments.
We will show that $N_\phi^*$ has exactly $1+2^n \cdot k$ linear regions according to \Cref{def:region4}, which implies the theorem.

We have $\{x\in \RR^n: T_{n,\delta}(x)>0\}\subseteq \bigcup_{z\in \{0,1\}^n}B^\infty_\delta(z)$ and $\{x\in \RR^n: f_{N_\phi}(x)>0\}\subseteq \bigcup_{z\in \{0,1\}^n}B^\infty_\varepsilon(z)$ by \Cref{lem:teps_regions} and \Cref{thm:np_completeness_2region_decision_problem}, respectively.
Since $\delta<\varepsilon$, we have $B^\infty_\delta(z)\subseteq B^\infty_\varepsilon(z)$ for every $z\in \{0,1\}^n$, and thus $\{x\in \RR^n: f_{N^*_\phi}(x)>0\}\subseteq \bigcup_{z\in \{0,1\}^n}B^\infty_\delta(z)$.
We now show that $f_{N^*_\phi}$ has $2^n$ nonzero regions in $B^\infty_\delta(z)$ if $z\in \{0,1\}^n$ satisfies $\phi$, and no nonzero regions $B^\infty_\delta(z)$ if $z$ does not satisfy $\phi$.

Suppose that $z\in \{0,1\}^n$ does not satisfy $\phi$.
Then, we have $f_{N_\phi}(x)=0$ for all $B^\infty_\varepsilon(z)$ by \Cref{thm:np_completeness_2region_decision_problem}, and thus $f_{N^*_\phi}$ has no nonzero region on $B^\infty_\delta(z)$.

Now suppose that $z\in \{0,1\}^n$ satisfies $\phi$.
By symmetry, we can assume without loss of generality that $z=(1,\dots, 1)$.
Since $z$ satisfies $\phi$, all clauses $i\in[m]$ must satisfy $|P_i|\geq 1$ and $|N_i|\leq n-1$.
Thus, for all $x\in B^\infty_\delta(z)=[1-\delta,1+\delta]^n$ and all $i\in [m]$
\begin{align*}
    1-\sum_{j\in P_i}x_j-\sum_{j\in N_i}(1-x_j)
    \leq 1-|P_i|(1-\delta)+|N_i|\delta
    \leq  \delta+(n-1)\delta
    = n \delta.
\end{align*}
Hence, for all $x\in [1-\delta,1+\delta]^n$, we have
\begin{align*}
    f_{N_\phi}(x)
    &\geq \max\left(0, T_n(x)+\varepsilon-\sum_{i=1}^m \max\left(0, 1-\sum_{j\in P_i}x_j-\sum_{j\in N_i}(1-x_j)\right)\right)\\
    &\geq \max\left(0, T_n(x)+\varepsilon-\sum_{i=1}^m n \delta\right)\\
    &= \max\left(0, T_n(x)+\varepsilon-nm \delta\right)\\
    &=\max\left(0, T_n(x)+\delta\right)
    = T_{n,\delta}(x),
\end{align*}
and thus, $f_{N_\phi^*}(x)=T_{n,\delta}(x)$ for all $x\in [1-\delta,1+\delta]^n$.
By \Cref{lem:teps_regions}, $T_{n,\varepsilon}$ has $2^n$ nonzero linear regions in $[1-\varepsilon,1+\varepsilon]^n$.
Thus, if $\phi$ has $k$ satisfying assignments, then $f_{N_\phi^*}$ has exactly $1+2^n\cdot k$ linear regions according to \Cref{def:region4} (one zero region and $2^n$ nonzero regions per satisfying assignment).
Finally, we can extend the hardness result to networks with $L\geq 3$ hidden layers as in the proof of \Cref{thm:np_hardness_Kregion_decision}.
\end{proof}
We note that \Cref{thm:sharp_sat_reduction} does not hold for \Cref{def:region5,def:region6}, since the constructed network has multiple open connected regions with the same affine function, some of which merge into a single closed connected region via low-dimensional segments.
\Cref{thm:sharp_sat_reduction} shows that a network having few regions does not necessarily imply that the regions of the network are ``easy to count''.
On the contrary, it shows that instances with relatively few regions can lead to \#\P-hard counting problems.

The reduction from \#\SAT{} implies that \emph{approximating} \textsc{Linear region counting} is intractable as well, in the following sense: We define an \emph{approximation algorithm} achieving approximation ratio $\rho \leq 1$ as an algorithm that is guaranteed to return, given a network $N$ with $R$ regions as input, a number $Q$ satisfying $\rho R \leq Q\leq R$. In fact, even though \Cref{thm:sharp_sat_reduction} only holds for \Cref{def:region4}, it is sufficient to prove the inapproximability result also for \Cref{def:region3,def:region5,def:region6}.
\begin{theorem}
\label{exponential_hardness_of_approximation_result}
    For any fixed constant $L\in \mathbb{N}, L\geq 3$, it is \NP-hard to approximate \textsc{Linear region counting} for \Cref{def:region3,def:region4,def:region5,def:region6} within an approximation ratio larger than $(2^n+1)^{-1}$ for networks with $L$ hidden layers and input dimension $n$.
\end{theorem}
\begin{proof}
If a \SAT\ formula has no satisfying assignment, the network produced by the reduction of \Cref{thm:sharp_sat_reduction} will be equal to zero and have exactly one linear region according to \Cref{def:region3,def:region4,def:region5,def:region6}.
Otherwise, it will have at least $1+2^n$ linear regions according to \Cref{def:region3,def:region4,def:region5,def:region6}. If it was possible to achieve an approximation ratio larger than $(2^n+1)^{-1}$ in polynomial time we could decide if a \SAT\ formula is satisfiable in polynomial time. This concludes the proof.   
\end{proof}
We note that very similar ideas also rule out a fully polynomial randomized approximation scheme (FPRAS) for approximating the number of linear regions.
A FPRAS \cite{jerrum2003counting} is a randomized polynomial-time (in the size of the input and $1/\varepsilon$) algorithm that returns a number $Q$ such that $\text{Prob}[(1-\varepsilon)R  \leq Q\leq (1+\varepsilon)R]\geq 3/4$,
where $R$ is the number of regions of the network.
\Cref{exponential_hardness_of_approximation_result} can be adapted to show no FPRAS exists for estimating the number of regions of neural network unless \NP\,=\,\RP.
To see this, observe that an unsatisfiable formula gives $R=1$, whereas a satisfiable one gives $R\geq 3$ in \Cref{exponential_hardness_of_approximation_result}.
Thus, a FPRAS with $\varepsilon=1/4$ would output, with probability at least $3/4$, a value $Q\leq 5/4$ in the unsatisfiable case, and $Q\geq 9/4$ in the satisfiable case, which would yield a polynomial-time randomized algorithm for SAT and imply \NP\,=\,\RP.

\Cref{exponential_hardness_of_approximation_result} does not imply hardness of approximation for ReLU networks with two hidden layers.
In the following, we show that approximation is already hard for networks with two hidden layers for \Cref{def:region3,def:region4,def:region5,def:region6}, although with a weaker inapproximability factor than in \Cref{exponential_hardness_of_approximation_result}.
\begin{theorem}
\label{thm:approx_result_2_hidden_layers}
Given a ReLU network with two hidden layers and input dimension $n$, for every $\varepsilon \in (0,1)$ it is \NP-hard to approximate \textsc{Linear region counting} by a ratio larger than $2^{-O(n^{1-\varepsilon})}$ for \Cref{def:region3,def:region4,def:region5,def:region6}.
\end{theorem}
\begin{proof}
Let $\phi$ be a \SAT{} formula on $\ell$ variables and let $N_\phi$ be the ReLU network with two hidden layers constructed in the proof of \Cref{thm:np_completeness_2region_decision_problem}.
Recall that for \Cref{def:region3,def:region4,def:region5,def:region6}, the network $N_\phi$ has at least two linear regions if $\phi$ is satisfiable and exactly one linear region if $\phi$ is unsatisfiable.

Let $N_\phi^{(k)}$ be the ReLU network composed of taking $k$ copies of $N_\phi$, each with a disjoint set of variables.
Write $y_i=(x_{i1},\dots, x_{i\ell})$ for $i\in [k]$.
The function computed by the ReLU network $N_\phi^{(k)}$ is $f_{N_\phi^{(k)}}:\mathbb{R}^{\ell k}\to \mathbb{R}$ with
\[
f_{N_\phi^{(k)}}(y_1,\dots,y_k)=\sum_{i=1}^k f_{N_\phi}(y_i).
\]
If $\phi$ is unsatisfiable, then $N_\phi$ has exactly one linear region, and thus $N_\phi^{(k)}$ also has exactly one linear region (according to \Cref{def:region3,def:region4,def:region5,def:region6}).
If $\phi$ is satisfiable, then $N_\phi$ has at least two affine regions, one zero affine region $X_0\subseteq \RR^\ell$, and one nonzero affine region $X_1\subseteq \RR^\ell$ on which the nonzero affine function $h:\RR^\ell\to \RR$ is computed.
For every $\sigma\in \{0,1\}^k$, the affine function $f_{N_\phi^{(k)}}(y_1,\dots,y_k)=\sum_{i: \sigma_i=1}h(y_i)$ is computed for all $(y_1,\dots, y_k)\in X_{\sigma_1}\times \dots \times X_{\sigma_k}\subseteq \mathbb{R}^{\ell k}$.
Since $X_{\sigma_1}\times \dots \times X_{\sigma_k}$ is full-dimensional and all these affine functions are different, $N_\phi^{(k)}$ has at least $2^k$ affine regions.
By \Cref{thm:region_hierarchy}, $N_\phi^{(k)}$ then also has at least $2^k$ linear regions according to \Cref{def:region3,def:region4,def:region5}.
It follows that approximating the number of regions within a factor larger than $2^{-k}$ is \NP-hard (according to \Cref{def:region3,def:region4,def:region5,def:region6}).
The theorem now follows by choosing $k=\ell^C$ for fixed constant $C$ with $C/(C+1)>1-\varepsilon$.
The input dimension is $n=\ell k = \ell^{C+1}$ and $k=n^{C/(C+1)}\geq n^{1-\varepsilon}$.
Since $C$ depends only on the fixed $\varepsilon$, $N_\phi^{(k)}$ can be constructed from $N_\phi$ in polynomial time.
\end{proof}

\section{Counting regions using polynomial space}\label{sec:pspace}
In this section, we first prove preliminary results that certain quantities can be bounded in and certain calculations can be carried out in polynomial time in \Cref{subsec:polynomial}, and then use these results to give polynomial space algorithms for counting the number of linear regions for \Cref{def:region1,def:region2,def:region6} in \Cref{subsec:PSPACEalgorithms}.
\subsection{Preliminary results}\label{subsec:polynomial}
We first give an explicit polynomial upper bound on the encoding size of coefficients in affine functions of affine regions of a network.
\begin{lemma}
	\label{lem:coeffs_are_polynomial}
	Let $N$ be a ReLU network, let $a\in \{0,1\}^{s(N)}$, and let $A_{\max}\in \mathbb{N}_{\geq 1}$ be the maximum absolute value of any numerator or denominator in a weight or bias of $N$.
	Then, every coefficient of the affine function $f_N^a:\mathbb{R}^n\to \mathbb{R}$, $x\mapsto c^\top x + b$ has encoding size at most
	\[
	4 (d+1)(n_{\max}+1)^2\langle A_{\max}\rangle.
	\]
\end{lemma}
\begin{proof}
	Given a vector $a\in \{0,1\}^{s(N)}$, we modify the matrices $A^{(1)},\dots, A^{(d+1)}$ and biases $b^{(1)},\dots, b^{(d+1)}$ by replacing the rows of matrices and bias entries corresponding to inactive neurons with 0 entries.
    More precisely, let $D_a^{(\ell)}\in \RR^{n_\ell\times n_\ell}$ be the diagonal matrix where all non-diagonal entries are zero, with $(D_a^{(\ell)})_{jj}=1$ if the $j$-th neuron of the $\ell$-th hidden layer is active in $a$.
    Set
    \[
    B^{(i)}=\begin{pmatrix}
        D_a^{(\ell)} A^{(\ell)} & D_a^{(\ell)} b^{(\ell)} \\
        0 & 1
    \end{pmatrix}
    \text{ for }\ell\in [d],
    \quad
    B^{(d+1)}=(A^{(d+1)},b^{(d+1)}).
    \]
    Then the row vector of coefficients of $f^a_N$ is 
	\[
	(c^\top, b) = B^{(d+1)} \cdots B^{(1)}.
	\]
    For every $\ell\in [d+1]$, let $p_\ell$ be the product of all denominators of $B^{(\ell)}$ and set $C^{(\ell)}=p_\ell B^{(\ell)}$.
    Since $B^{(\ell)}$ has at most $(n_{\max}+1)^2$ entries, we have
    \[
    p_\ell \leq A_{\max}^{(n_{\max}+1)^2},
    \qquad
    |(C^{(i)})_{kl}|\leq  A_{\max}^{(n_{\max}+1)^2}\text{ for all }k,l.
    \]
    Since $(c^\top, b) = \frac{1}{\prod_{\ell=1}^{d+1}p_\ell}C^{(d+1)} \cdots C^{(1)}$, the denominator of any entry of $(c^\top, b)$ has absolute value at most $\prod_{\ell=1}^{d+1}p_\ell\leq A_{\max}^{(d+1)(n_{\max}+1)^2}$.
    Further, the numerator of this entry is at most the absolute value of an entry of $C^{(d+1)} \cdots  C^{(1)}$, which is a sum of at most $(n_{\max}+1)^d$ many products of $d+1$ terms, so the numerator has absolute value at most $(n_{\max}+1)^dA_{\max}^{(d+1)(n_{\max}+1)^2}$.
    Hence, using that $\langle a\cdot b\rangle \leq \langle a\rangle + \langle b\rangle$ for all $a,b\in \mathbb{N}$, we obtain that any entry of $(c^\top, b)$ has encoding size at most
    \begin{align*}
    &1+ \langle (n_{\max}+1)^dA_{\max}^{(d+1)(n_{\max}+1)^2}\rangle+ \langle A_{\max}^{(d+1)(n_{\max}+1)^2}\rangle\\
    \leq\;
    &1+ d\langle n_{\max}+1\rangle + 2 (d+1)(n_{\max}+1)^2\langle A_{\max}\rangle\\
    \leq\;
    &4 (d+1)(n_{\max}+1)^2\langle A_{\max}\rangle.\qedhere
	\end{align*}
\end{proof}
We show that for a specified activation pattern, both the affine preactivation function and the affine output function of a neuron can be computed in polynomial time.
\begin{lemma}
	\label{lem:pattern_to_function_is_poly_time_computable}
	Given a ReLU network $N$, a vector $a\in \{0,1\}^{s(N)}$, and an index $i\in [s(N)]$ of a neuron, one can compute in time polynomial in $\langle N\rangle$ the coefficients of both affine functions $g^a_i$ and $f^a_i$.
\end{lemma}
\begin{proof}
    The same matrix-product formula as in the proof of \Cref{lem:coeffs_are_polynomial} shows that the coefficients of all output functions can be obtained in polynomial time.
    To see this, observe that for each hidden layer $\ell \in [d]$, the matrix product $B^{(d+1)}\cdots B^{(1)}$ gives the coefficients of the network output $f^a_N$, and by truncating this construction after any layer one obtains the coefficients of the output functions $f_i^a$ of all neurons in that layer.

    Once the output functions of the previous layer are known, the preactivation functions of a neuron can be computed by taking the corresponding affine combination and adding the bias.
    More precisely, if neuron $i$ lies in layer $\ell$, then
    \[
    g_i^a(x)=\sum_{j=1}^{n_\ell -1}A^{(\ell)}_{ij}f_j^a(x)+b_i^{(\ell)},
    \]
    where we write $(f_1^a(x),\dots,f_{n_0}^a(x))=(x_1,\dots,x_{n_0})$ for $\ell=1$.
    Therefore, the coefficients of $g_i^a$ can also be computed in polynomial time.
\end{proof}
Combining \Cref{lem:coeffs_are_polynomial,lem:pattern_to_function_is_poly_time_computable}, we show that the dimension of $S_a$ for a given vector $a\in \{0,1\}^{s(N)}$ can be computed in polynomial time.
\begin{lemma}
	\label{lem:dim_of_activation_region_is_poly_time_computable}
	Given a ReLU network $N$ and a vector $a\in \{0,1\}^{s(N)}$ with support $I\subseteq [s(N)]$, one can compute the dimension of $S_a=\{x\in \mathbb{R}^n: g_i^{a}(x)>0\; \text{ for all } i\in I,\; g_i^{a}(x) \leq 0\; \text{ for all } i\notin I\}$ in time polynomial in $\langle N\rangle$.
\end{lemma}
\begin{proof}
    For each neuron, let $g^a_i:\RR^n\to \RR$ be the affine preactivation function obtained from \Cref{lem:pattern_to_function_is_poly_time_computable}.
	By \Cref{lem:coeffs_are_polynomial}, the encoding size of the coefficients of these functions is bounded polynomially in $\langle N\rangle$.
    Consider the polyhedron
	\[
	P_a = \{x\in \mathbb{R}^n: g_i^a(x)\geq 0\; \text{ for all } i\in I,\; g_i^a(x) \leq 0\; \text{ for all } i\notin I\}.
	\]
    By solving a series of linear programs, we can compute the dimension of $P_a$ (for details, see \citep[Section 7.3]{20.500.11850/426218}).
    If $P_a$ is empty, then $S_a$ is also empty (since $S_a\subseteq P_a$).
    Suppose now that $P_a$ is nonempty.
    Observe that, in this case, $S_a$ is nonempty exactly when, for every $i\in I$, the equality $g_i^a(x)=0$ does not hold for all $x\in P_a$.
    To see this, note that if for each $i\in I$, there is a point $x^{(i)}\in P_a$ with $g_i^a(x^{(i)})>0$, then $\frac{1}{|I|}\sum_{i\in I}x^{(i)}\in S_a$.
    In particular, one can decide whether $S_a$ is nonempty by solving the linear programs $\max\{g_i^a(x) : x\in P_a\}$ for all $i\in I$.
    Suppose now that $S_a$ is nonempty with $y\in S_a$.
    In this case, $P_a$ is the closure of $S_a$, and thus $\dim P_a=\dim S_a$.
    Hence, we can compute $\dim S_a$ in polynomial time.
\end{proof}
\subsection{Algorithms for counting regions}\label{subsec:PSPACEalgorithms}
Due to the \#\P-hardness of \textsc{Linear region counting}, we do not expect that efficient (polynomial time) algorithms for counting the number of linear regions exist.
\citet{wang2022estimation} claimed that \textsc{Linear region counting} for \Cref{def:region5} would be in \EXP\TIME.
As stated in \Cref{sec:flaws}, the algorithm actually works for \Cref{def:region4} instead of \Cref{def:region5}.
\citet{lezeau2024tropical} give an \EXP\TIME{} algorithm for computing linear regions for \Cref{def:region4} and remark that their algorithm can be adapted to a (finite) algorithm for \Cref{def:region5} \citep[Remark 16]{lezeau2024tropical}.
For \Cref{def:region3}, to the best of our knowledge, it is not clear whether an \EXP\TIME{} algorithm exists, because there are infinitely many options to choose the underlying polyhedral complex.
On the contrary, we show in this section that for \Cref{def:region1,def:region2,def:region6}, the number of regions can be computed in polynomial space and therefore also in \EXP\TIME{}.

While it is not hard to see that the number of activation regions and proper activation regions can be computed in polynomial space, counting the number of affine regions is slightly more complicated: a naive approach enumerating all the activation patterns would need to keep track of all the affine coefficients already seen to avoid double-counting, which is infeasible in polynomial space.
Instead, we iterate over all possible affine functions by enumerating all possible coefficient combinations that have an encoding size of less than a polynomial upper bound, and check if there is a proper activation region on which the affine function is realized (see \Cref{alg:exhaustive_search} for a formal description).
The running time of \Cref{alg:exhaustive_search} is exponential in the encoding size of the network, but it suffices to prove \FP\SPACE{} containment.
The comments in \Cref{alg:search_affine_piece,alg:exhaustive_search} refer to the lemmas showing that the computation in the respective line can be done in polynomial space.

\begin{algorithm}
\caption{\textsc{SearchAffinePiece}}
\label{alg:search_affine_piece}
    \begin{algorithmic}[1]
        \Require{A ReLU network $N$ and a vector $(w_1,\dots,w_n,b)\in \mathbb{Q}^{n+1}$.}
        \Ensure{1 if $\sum_{i=1}^n w_i x_i + b$ is a function of an affine region of $N$, else 0.}
        \For{$a\in \{0,1\}^{s(N)}$}
        \If{$\dim S_a=n$}\algorithmiccomment{(\Cref{lem:dim_of_activation_region_is_poly_time_computable})}
        \If{$\sum_{i=1}^n w_i x_i + b=f^a_{N}(x)$ }
        \Return 1 \algorithmiccomment{(\Cref{lem:pattern_to_function_is_poly_time_computable})}
        \EndIf
        \EndIf
        \EndFor
        \Return 0
    \end{algorithmic}
\end{algorithm}
\begin{algorithm}
\caption{\textsc{ExhaustiveSearch}}
\label{alg:exhaustive_search}
    \begin{algorithmic}[1]
        \Require{A ReLU network $N$.}
        \Ensure{Number of affine regions of $N$.}
        \State $n_{\max}=\max\{n_0,n_1,\dots, n_{d+1}\}$
        \State $U=2^{4 (d+1)(n_{\max}+1)^2\langle A_{\max}\rangle}$ \algorithmiccomment{(\Cref{lem:coeffs_are_polynomial})}
        \State $R=0$
        \For{$(w, b)\in \{-U,\dots, U\}^{n+1}\times \{1,\dots, U\}^{n+1}$}
        \If{$\gcd(w_i,b_i)=1$ for $i\in [n+1]$}
        \State $R \leftarrow R+\textsc{SearchAffinePiece}(N,(\frac{a_1}{b_1},\dots, \frac{a_{n+1}}{b_{n+1}}))$ \algorithmiccomment{(\Cref{alg:search_affine_piece})}
        \EndIf
        \EndFor
        \Return $R$
    \end{algorithmic}
\end{algorithm}

\begin{theorem}
\label{thm:poly_space_algorithms}
    \textsc{Linear region counting} is in \FP\SPACE{} for \Cref{def:region1,def:region2,def:region6}.
\end{theorem}
\begin{proof}
For \Cref{def:region1,def:region2}, one can simply iterate over all $2^{s(N)}$ vectors in $\{0,1\}^{s(N)}$, and for each vector $a\in \{0,1\}^{s(N)}$, compute the dimension of $S_a$ in polynomial time by \Cref{lem:dim_of_activation_region_is_poly_time_computable} to determine if $S_a$ is an activation region (if $\dim S_a \neq -1$) or a proper activation region (if $\dim S_a=n$), and increase a counter by one if this is the case.

For affine regions, set $M:=4 (d+1)(n_{\max}+1)^2\langle A_{\max}\rangle$.
Every affine function realized on a proper activation region has $n+1$ rational coefficients whose reduced numerators and denominators have absolute value at most $U=2^M$ by \Cref{lem:coeffs_are_polynomial} (line 2 in \Cref{alg:exhaustive_search}).
We then exhaustively search through all such reduced coefficient tuples (line 4 and 5 in \Cref{alg:exhaustive_search}): for every such tuple, \textsc{SearchAffinePiece} determines in polynomial space whether the corresponding affine function is the function of an affine region of $N$.
In \textsc{SearchAffinePiece}, we use the fact that an affine function is realized by an affine region if and only if it is realized by at least one proper activation region and thus iterate over all activation patterns and return $1$ if the activation region corresponding to the activation pattern is proper (by \Cref{lem:dim_of_activation_region_is_poly_time_computable}) and realizes the affine function (which we can check in polynomial time by \Cref{lem:pattern_to_function_is_poly_time_computable}).
If an affine function of an affine region is found, we increase a counter by 1.
To avoid counting the same affine function more than once, we only check tuples of reduced fractions.

Since the number of candidate tuples is at most $2^{\operatorname{poly}(\langle N\rangle)}$, the counter needed for iteration has polynomial bit length.
Hence, the entire computation uses polynomial space.
\end{proof}

The algorithms for counting (proper) activation regions (\Cref{def:region1,def:region2}) show fixed-parameter tractability of \textsc{Linear region counting} with respect to the number of neurons.
Furthermore, recent results of \cite[Corollary 4.4]{froese2025parameterized} imply W[1]-hardness of \textsc{Linear region counting} for ReLU networks with two or more hidden layers when parameterized by the input dimension (for \Cref{def:region3,def:region4,def:region5,def:region6}).
However, the fixed-parameter tractability status of \textsc{Linear region counting} remains open for other parameterizations and definitions.

\section{Conclusion}
We collected and discussed six commonly used non-equivalent definitions for linear regions of ReLU networks.
We proved \#\P-hardness for counting the number of linear regions (for all six definitions) and \NP-hardness for several associated decision problems (for most definitions).
We further showed that for ReLU networks with two or more hidden layers, even approximating the number of linear regions is \NP-hard (again, for most definitions).
On the positive side, we showed that for some definitions, linear regions can at least be counted in polynomial space.

There remain many interesting open problems and directions for future work.
Is \textsc{Linear region counting} in \EXP\TIME\ for \Cref{def:region3} (are there even finite algorithms)?
Can some of the results in \Cref{sec:shallow,sec:deep,sec:pspace} be extended to all six definitions?
For example, is \textsc{Linear region counting} for ReLU networks with one hidden layer contained in \#\P\ also for \Cref{def:region5,def:region6}?
Is the problem of approximating the number of linear regions also \NP-hard for ReLU networks with one hidden layer?
Further, it would be interesting to devise algorithms for region counting on networks that have been trained using gradient descent, as there is evidence that such networks have fewer regions~\cite{hanin2019deep}, which might allow faster algorithms.
Finally, it would be interesting to study the fixed-parameter tractability of \textsc{Linear region counting} under different parameterizations and definitions.

\subsubsection*{AI disclosure}
We used large language models (LLMs), specifically OpenAI's ChatGPT 5.6, to assist with creating TikZ figures, provide suggestions for improving the readability of the paper, and for proofreading.
In the proofreading process, ChatGPT pointed out an error in Lemma A.1 in the previous version of the paper \cite{stargalla2026computational}, which claimed that the set of open connected regions is equal to the unique set of open connected subsets satisfying certain assumptions, and proposed fixing the statement by restricting to regular open subsets.
All other mathematical claims in this paper were developed by the authors and not generated by LLMs.
The authors reviewed the AI-assisted material and take full responsibility for the content of the paper.

\bibliographystyle{plainnat}
\bibliography{arxiv.bib}
\end{document}